%% file: main.tex
\newif\ifdouble
\doublefalse

\ifdouble
\documentclass[conference,a4paper]{IEEEtran}

\else
\documentclass[journal,12pt,onecolumn]{IEEEtran}
\usepackage[letterpaper, top=1in, bottom=1in, left=1in, right=1in]{geometry}

\usepackage{setspace}
\fi

\usepackage{footnote}

\usepackage[utf8]{inputenc}
\usepackage[T1]{fontenc}
\usepackage{url}
\usepackage[ruled,vlined]{algorithm2e}
\usepackage{ifthen}
\usepackage{cite}
\usepackage{graphicx}
\usepackage[cmex10]{amsmath}
\usepackage{comment}
\usepackage{amsthm}
\usepackage{amsmath,amssymb,amsfonts}
\usepackage{nicefrac}
\usepackage{graphicx}
\usepackage{algorithmic}
\usepackage{bm}
\usepackage{cleveref}
\usepackage{graphicx}
\usepackage{subcaption}

\newtheorem{theorem}{Theorem}
\newtheorem{lemma}{Lemma}
\newtheorem{definition}{Definition}
\newtheorem{corollary}{Corollary}
\newtheorem{example}{Example}

\newtheorem{assum}{Assumption}
\newtheorem{opt}{Optimization Problem}
\newcommand{\paren}[1]{\left( #1 \right)}

\newcommand{\argmin}[1]{\underset{#1}{\arg \min}}
\newcommand{\st}{\text{s.t.}}
\newcommand{\argmax}[1]{\underset{#1}{\arg \max}}

\usepackage{mathtools}

\newcommand{\off}[1]{}

\usepackage{xcolor}

\newcommand{\rcomp}{r_p^{\text{comp}}}

\newcommand{\rcomm}{r_p^{\text{comm}}}
\newcommand{\phil}{\underline{\phi}_p}

\definecolor{myblue}{rgb}{0.01, 0.28, 1}
\definecolor{mygreen}{rgb}{0.13, 0.55, 0.13}

\usepackage{tikz}     % Un package pour les dessins (utilisé pour l'environnement {code})
\usetikzlibrary{calc,math}
\usepackage[framemethod=TikZ]{mdframed}
\usepackage{pgfplots}
\pgfplotsset{compat=1.16}
\usepackage{pifont}% http://ctan.org/pkg/pifont
\definecolor{DarkGreen}{rgb}{0.1,0.5,0.1}
\definecolor{DarkRed}{rgb}{0.5,0.1,0.1}
\definecolor{DarkBlue}{rgb}{0.1,0.1,0.5}
\definecolor{DarkPurple}{rgb}{0.5,0.2,0.5}
\definecolor{DarkTurquoise}{rgb}{0.1,0.5,0.5}

\newif\ifcomen
\comentrue
\begin{document}

\title{Stream Distributed Coded Computing}

\author{}

\author{%
   \IEEEauthorblockN{\hspace{-0.1cm}Alejandro Cohen\IEEEauthorrefmark{1},
                     Guillaume Thiran\IEEEauthorrefmark{2},
                     Homa Esfahanizadeh\IEEEauthorrefmark{1},
                     and Muriel M\'edard\IEEEauthorrefmark{1}}\\
   \IEEEauthorblockA{\IEEEauthorrefmark{1}%
                     RLE, MIT, Cambridge, MA, USA, \{cohenale, homaesf, medard\}@mit.edu}\\
   \IEEEauthorblockA{\IEEEauthorrefmark{2}%
                     UCLouvain, Belgium,  guillaume.thiran@uclouvain.be\footnote{{ }Authors have equal contributions.}}\\
}

\maketitle

\begin{abstract}
The emerging large-scale and data-hungry algorithms require the computations to be delegated from a central server to several worker nodes. One major challenge in the distributed computations is to tackle delays and failures caused by the stragglers. To address this challenge, introducing efficient amount of redundant computations via distributed coded computation has received significant attention. Recent approaches in this area have mainly focused on introducing minimum computational redundancies to tolerate certain number of stragglers. To the best of our knowledge, the current literature lacks a unified end-to-end design in a heterogeneous setting where the workers can vary in their computation and communication capabilities. The contribution of this paper is to devise a novel framework for joint scheduling-coding, in a setting where the workers and the arrival of stream computational jobs are based on stochastic models. In our initial joint scheme, we propose a systematic framework that illustrates how to select a set of workers and how to split the computational load among the selected workers based on their differences in order to minimize the average in-order job execution delay. Through simulations, we demonstrate that the performance of our framework is dramatically better than the performance of naive method that splits the computational load uniformly among the workers, and it is close to the ideal performance.
\end{abstract}

\begin{IEEEkeywords}
Distributed coded computation, stragglers, large matrix-matrix multiplication, large matrix-vector multiplication, ultra-reliable low-latency, in-order execution delay, queuing theory.
\end{IEEEkeywords}

%%%%%%%%%%%%%%%%%%%%%%%%%%%%%%%%%%%%%%%%%%%%%%%%%%%%%%%%%%
\section{Introduction}
%%%%%%%%%%%%%%%%%%%%%%%%%%%%%%%%%%%%%%%%%%%%%%%%%%%%%%%%%%

\begin{figure}
    \centering
    \ifdouble
       \includegraphics[width = 1 \columnwidth]{non_and_iterative_sys_model_1}
    \else
       \includegraphics[width = 0.95 \columnwidth]{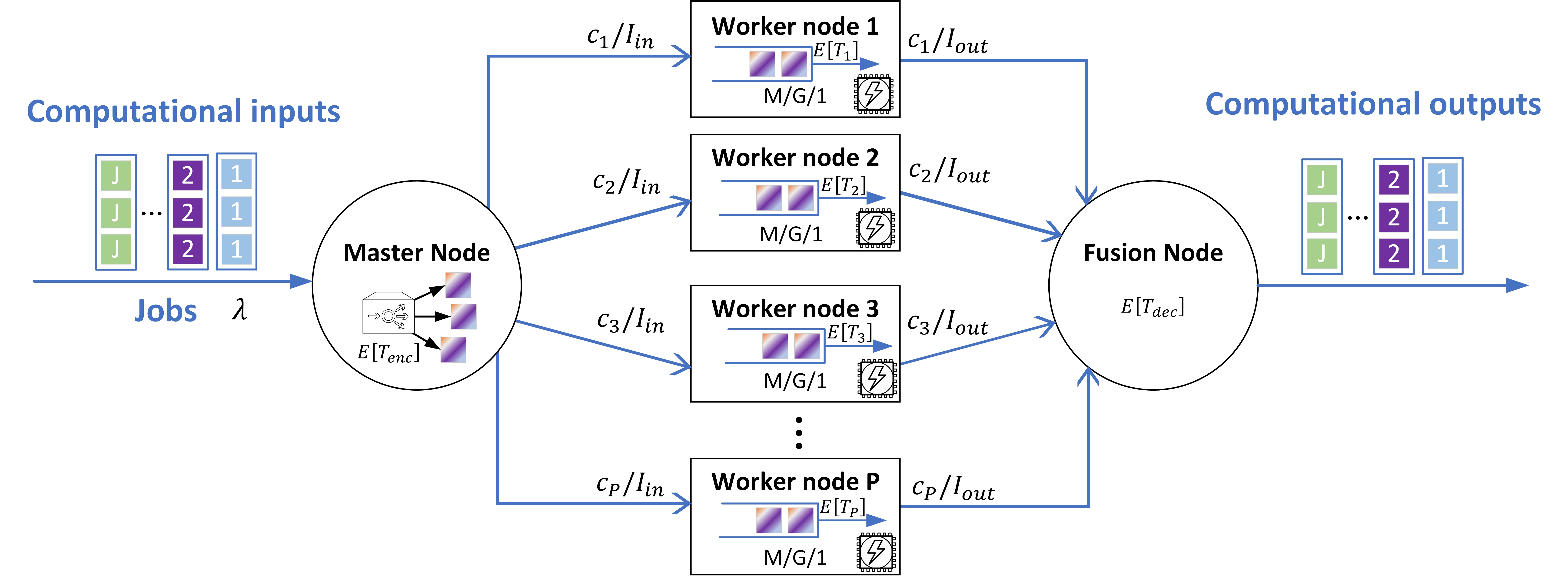}
    \fi
    \captionof{figure}{{System model for the distributed computation problem. The parameters are: $\lambda$ for the average job arrival rate. $E[T_\text{enc}]$ and $E[T_\text{dec}]$ for the average encoding time and decoding time, respectively. $E[T_p]$, $c_p/I_\text{in}$, and $c_p/I_\text{out}$ for the average computing time, incoming communication time, and outgoing communication time per job, respectively. Each worker is modeled with an M/G/1 queue.}}
    \label{fig:system_model_2}
\end{figure}

We are in the era of big data where the size and dimension of data grow dramatically surpassing the Moore's law \cite{moore1965cramming,moore1965moore}. As a result, the central heavy computations over massive amount of data have become more unrealistic, and distributing the computations over several \textit{workers} by a \textit{master node} has become the rising solution \cite{dutta2019short,mallick2019rateless,yu2019lagrange}\footnote{In another view of distributed computation, which is not the focus of this paper, there is no master to orchestrate, and the workers coordinate among themselves.}. The distributed computation requires a computational \textit{job} to be split into several \textit{tasks} with lower complexity and smaller input data, such that they can be distributed among a set of workers and the final job result can be obtained by a \textit{fusion node} using the combined task results. Fig.~\ref{fig:system_model_2} illustrates a system model for the distributed computation problem.

One major challenge in the distributed computation framework is addressing delays and failures caused by slow and/or unreliable workers known as \textit{stragglers} affecting the overall performance of the system \cite{dean2013tail,joshi2017efficient,ramamoorthy2020straggler}. The efforts to reduce the bottleneck effect of stragglers has opened up a new line of research as distributed coded computation that has attracted significant attention \cite{sheth2018application,d2020gasp,wahabfederated}. The idea is to introduce redundancies in the distributed computations such that the output can be retrieved from an arbitrary subset of task results, and so the performance of the system will not be limited by the stragglers \cite{yu2017polynomialn,raviv2020gradient}.

The majority of previous works on distributed coded computation focus on decoding the result of a computational job given a subset of successful task results assuming a homogeneous cluster of workers \cite{mallick2019rateless}. In practice, there is a variance among computational powers of the workers depending on their CPU/GPU powers, aging factors, maintenance conditions, etc \cite{sarikaya2019motivating,lim2020federated}. Besides, the communication links to and from each worker can have various rates causing various communication latency \cite{mallick2019rateless,feng2019joint,9155442}. Incorporating the heterogeneity of the workers into the distributed coded computation design has been less studied in the literature. In \cite{duffy2021mds}, a model is considered where maximum-distance-separable (MDS) codes{\cite{li2016unified,lee2017high,kim2019coded,lee2017speeding}} are used for generating computational redundancies, and redundant computations are removed from queue of workers once enough tasks are finished. It is shown via a precise queuing analysis that replication codes result in a higher delay compared to the MDS codes in this setting. However, a scheduling that incorporates the variabilities at the workers into the design is an open challenge.

\subsection*{Main Contribution}
In this paper, we propose a new adaptive framework to execute a stream of non-iterative computational jobs with low in-order execution delay in a distributed fashion, see Fig.~\ref{fig:system_model_2}. We consider each worker is modeled with an M/G/1 queue\footnote{In queuing theory, an M/G/1 queue is a model with a single processor where arrivals have Poisson distribution and service times have a General distribution.}, and the workers can vary in their service rates. This is a first step to design a joint scheduling-coding framework for stream distributed computation utilizing queuing preliminaries. The average job execution delay depends on the system model parameters, parameters of the coding scheme for splitting the jobs into smaller tasks, and scheduling scheme for distribution of the tasks among the acquired workers. The scheduling-coding optimization is critical to minimize the in-order execution delay of the jobs and to maximize the utilization of the workers.

To this end, the proposed framework selects optimal code parameters and optimally splits the computations among the workers based on their computational and communication power. Despite the previous works which focus on one or few targets, we consider all sources of delays in our end-to-end system design, i.e., encoding/decoding delay, queuing delay, computation delay, and communication delay, to optimize the parameters. Our solution provides a balance between the end-to-end in-order execution delay and the total acquired computational resources for stream computing applications.  Our framework is compatible with a general form of distributed computation and can be combined with many existing coded computation codes, e.g., \cite{dutta2019short,tandon2016gradient,tandon2017gradient,raviv2018gradient,Dutta2020,sheth2018application}, to optimize their code parameters in order to reduce the in-order execution delay and improve their utilization.

Our method can be briefly described as follows. We assume a possibly large set of workers are available to select from. We first choose an appropriate set of workers that provides sufficient resources and keeps the system stable. Then, the optimal load split (contribution portion of each worker) is identified to minimize the in-order execution delay. This optimization process is repeated periodically during the run-time, or when changes occur in the system\footnote{The memory-less job arrival model at the workers, i.e., M/G/1 model, is a result of the memory-less job arrival model at the master node. This model leads us to provide an optimal load split using the various stochastic features of the service time at the workers. Considering a general job arrival distribution at the master node requires more complex considerations in order to formulate the average job execution delay of the system. We leave this interesting extension as future work.}. We note that the optimal load split depends on first and second moment of each worker's service rate, and can be obtained via feedback and/or via prior knowledge of the master node.

The rest of this paper is organized as follows. In Section~\ref{sim:related}, we summarize the related works. In Section~\ref{sec:prob}, we describe the system model. In Section~\ref{sec:sol}, we propose our new framework for joint scheduling-coding in a distributed computation problem. In Section~\ref{sec:feedback}, we describe how to adaptively adjust the parameters needed for our proposed optimal solution using feedback. Sections~\ref{sec:sim} and \ref{sim:conc} are dedicated to simulation results and conclusions, respectively.

%%%%%%%%%%%%%%%%%%%%%%%%%%%%%%%%%%%%%%%%%%%%%%%%%%%%%%%%%%
\section{Related Work}\label{sim:related}
%%%%%%%%%%%%%%%%%%%%%%%%%%%%%%%%%%%%%%%%%%%%%%%%%%%%%%%%%%

\begin{figure}
    \centering
    \includegraphics[width = 1 \columnwidth]{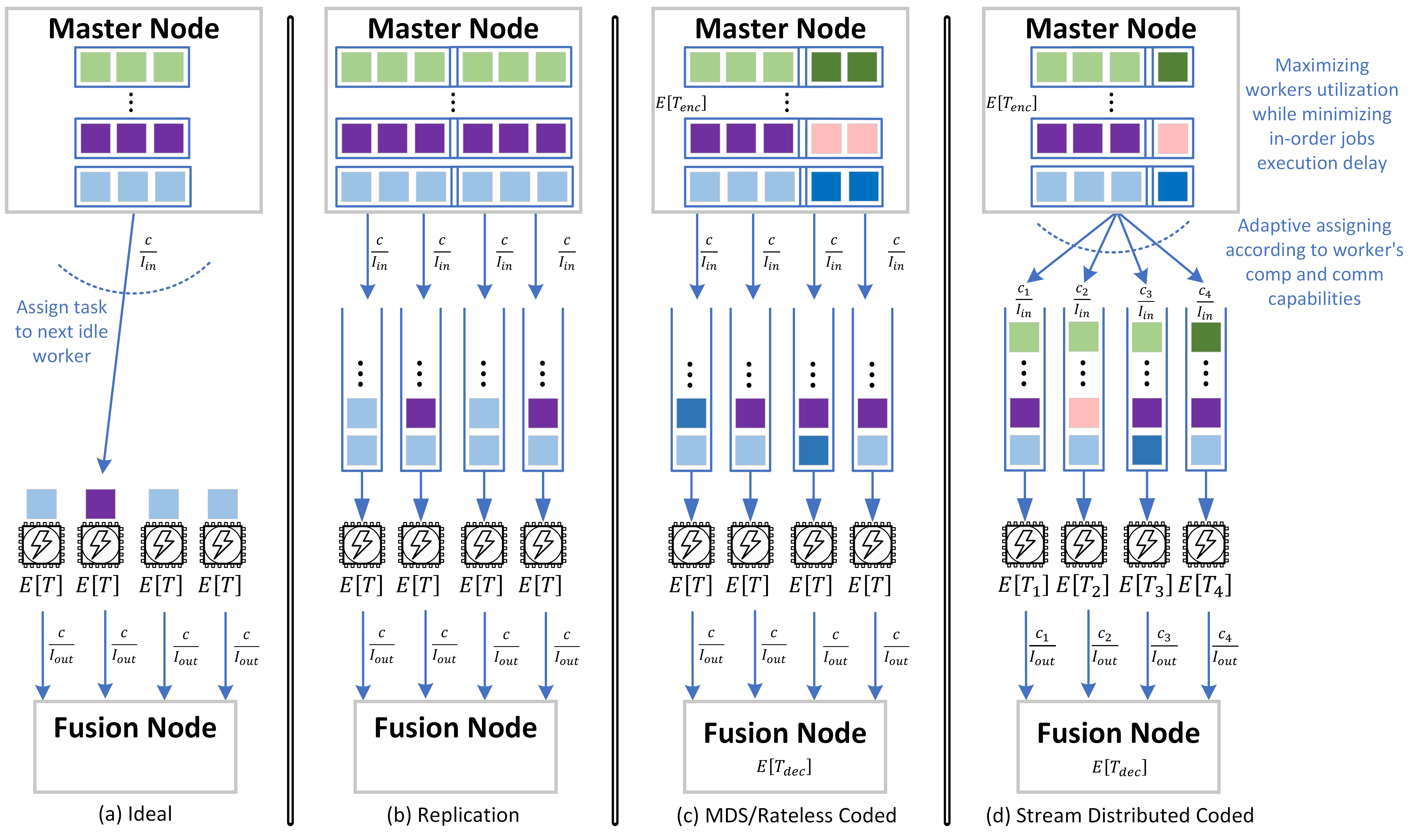}
    \captionof{figure}{{Various distributed computational schemes. Each computational job is split to several smaller tasks. The tasks related to one job are marked with the same color, and the redundant computations are distinguished from the necessary ones with a lighter color.}}
    \label{fig:code_ex}
\end{figure}

How to introduce computational redundancies in distributed computation has been considered in many modern applications that require heavy computations over a large amount of data. An interesting example is multiplication of two large matrices \cite{yu2017polynomial,Dutta2020,lee2017high,lee2017speeding,baharav2018straggler,li2016unified} and multiplication of a large matrix and a vector \cite{dutta2019short,mallick2019rateless}. Another example is MapReduce where a computationally intensive job of identifying specific features over a possibly large dataset, e.g., specific word counts in a large file,
is distributed among several workers such that each worker performs the process over one or several smaller subsets of the dataset \cite{li2015coded}. Random shuffling of data is another attractive example which benefits from the coded computation and enhances the performance of large-scale machine learning algorithms \cite{attia2019near,song2019pliable}. There, work mainly focuses on optimizing the code parameters to pursue goals such as reducing the number of task results required to finish a job \cite{Dutta2020,yu2017polynomial}, providing a balance between computation load and communication load\cite{li2016unified}, reducing the
encoding and decoding complexities \cite{baharav2018straggler}, improving the communication efficiency\cite{attia2019near,song2019pliable}, among others. Recently, an information theoretic model of functional compression for distributed computation has been studied in \cite{9155442}.

Fig.~\ref{fig:code_ex} depicts various scheduling-coding methods for the distributed coded computation problem. (a) is the ideal solution, where in-order the jobs received at the master node are assigned to the workers as soon as they become idle. (b) is the replication solution, where the tasks for each job are replicated and then uniformly distributed among the workers. (c) is coding solution, where the tasks, being the result of a distributed coding scheme, eg., MDS codes\cite{li2016unified,lee2017high} and rateless codes \cite{mallick2019rateless} and including some redundant tasks, are generated and assigned to the workers. (d) is our proposed joint coding and scheduling framework that incorporates the diversities of the workers to optimally split the task load of each job among the workers. According to our simulation results, the optimal distribution of computational jobs according to our framework not only results in an average in-order execution delay much lower than the uniform job distribution, but it also is capable to provide stability in cases that the uniform split is not able to do so.

%%%%%%%%%%%%%%%%%%%%%%%%%%%%%%%%%%%%%%%%%%%%%%%%%%%%%%%%%%
\section{System Model}\label{sec:prob}
%%%%%%%%%%%%%%%%%%%%%%%%%%%%%%%%%%%%%%%%%%%%%%%%%%%%%%%%%%

Fig.~\ref{fig:system_model_2} illustrates the setting for distributed computational of a stream of computational jobs with latency constraints. The master node sequentially receives the inputs for several computational jobs. We model the job arrival at the master with a Poisson stochastic random process $\mathcal{N}(l)$ with parameter $\lambda$ that represents the number of jobs arriving up to and including time step $l$, $l\in\{1,\dots,L\}$,
\begin{equation*}
    P(\mathcal{N}(l)= n) = e^{-\lambda l}\frac{(\lambda l)^{n}}{n!}.
\end{equation*}
Thus, $\lambda$ also denotes the average job arrival rate at each time step, and we assume {$J=\mathcal{N}(L)$}.

Upon arrival of each job, the master node assigns appropriate \textit{portion} of the computational job, i.e., appropriate number of tasks, to each worker in a selected set of workers $\mathcal{P}$, where $|\mathcal{P}|=P$. The queue at each worker is modeled with an M/G/1 queue \cite{Klein1975}. Let $\phi_p$ be the portion of a job that is assigned to the $p$-th worker (which will be designed by our solution), $p\in \{1,\dots,P\}$. These $P$ workers are selected from a large set of workers by the proposed framework. Each worker independently performs pre-defined computational operations and sends the results to a fusion node where the received computational results from workers are aggregated to produce the desired computational output for each job, in-order. We denote as $I_\text{in}$, $I_\text{out}$, and $C$ the input size (number of symbols), output size (number of symbols), and computational complexity (number of operations) for each job.

Each worker has limited computational resources and communication capacity at each time step. That is, each worker can perform operations on and transmit a limited amount of computational inputs at a time. Besides, the computational speed and communication rate can vary from worker to worker.
We denote as $c_p$ the transmission rate of the $p$-th worker, i.e., the number of symbols that can be communicated to or from the $p$-th worker at each time step. Let $T_{p}$ denote the time it takes for $p$-th worker to perform one entire computational job, i.e., $C$ operations. Then $E[T_{p}]$ and $E[T_{p}^2]$ denote the first moment (average) and the second moment of the service time for an entire job for the $p$-th worker, respectively, and the service rate of the worker for entire job is $1/E[T_{p}]$. The master node has access to $E[T_p]$ and $E[T_p^2]$ according to the workers' declaration and/or estimations during the execution, which will be explained thouroughly in Section~\ref{sec:worker_estimation}

We assume the first moment and second moment of the time it takes for the $p$-th worker to finish $\phi_p$ portion of a job is $\phi_p E[T_p]$ and $\phi_p^2E[T_p^2]$, respectively. We define the \textit{scaled service rate} of the $p$-th worker as $\rcomp \triangleq {1}/{E[T_p]\lambda}$, which denotes the portion a job that the $p$-th worker is able to perform on average during ${1}/{\lambda}$ time steps. We also define the \textit{scaled communication rate} of the $p$-th worker as $\rcomm\triangleq{ c_p}/{(I_\text{in}+I_\text{out})\lambda}$, which denotes the portion a job that the $p$-th worker is able to transfer during ${1}/{\lambda}$ time steps. We remind that $\phi_p$ is the portion of one job that is assigned to the $p$-th worker by the designed solution and $\rcomp$ is the portion one job that the $p$-th worker is able to perform on average during ${1}/{\lambda}$ time steps according to its declaration and/or the master node's estimation. Then, the average time it takes for the system to compute a job is $
\frac{1}{P}\sum_{p=1}^{P}\phi_pE[T_p]$, and
\begin{equation*}
    \frac{1}{P}\sum_{p=1}^P\phi_p E[T_{p}] \leq \frac{1}{P}\sum_{p=1}^P\rcomp E[T_{p}],
\end{equation*}
where we require for stability that the time it takes to solve a job is lower than the inverse of the arrival rate.
Now, we consider the formulation of the computational jobs. Let consider the $j$-th computational job, $j\in\{1,\dots,J\}$, arrived at the master node as
\begin{equation}\label{equ:distribyted-job}
    f\left(X_1(j),X_2(j),...,X_m(j)\right),
\end{equation}
where $X_1(j),X_2(j),...,X_m(j)$ are the computational inputs for the $j$-th job, and $f(.)$ describes the computational job. Each job can be divided into several smaller computational tasks that has smaller inputs and lower computational complexity. Each computational task works toward completion of a single job, the computational tasks work on inputs of the same size, and the input of a computational task does not have to be one of the computational inputs (it can be a function of them). We assume each task can be performed by one worker and requires a specific number of operations. The fusion node requires to obtain $K$ completed task results, each with the same size, per job to identify the final result.

\begin{table}
    \centering
    \begin{tabular}{|l|l|}
    \hline
    Parameter & Definition\\
    \hline
    \hline
    $l$&time step index, $l\in\{1,\dots,L\}$\\
    \hline
    $\lambda$ & average job arrival rate\\
    \hline
    $\mathcal{P}$ and $P$& set of acquired workers and its cardinality\\
    \hline
    $T_p$ & job service time of the $p$-th worker\\
    \hline
    $\rcomp$&scaled service rate of the $p$-th worker\\
    \hline
    $\rcomm$&scaled communication rate of the $p$-th worker\\
    \hline
    $\phi_p$ & fraction of load on the $p$-th worker\\
    \hline
    $\mu_p$ & computing rate of operations for the $p$-th worker\\
    \hline
    $K$&number of critical tasks per job\\
    \hline
    $\Omega$&redundancy ratio\\
    \hline
    \end{tabular}
    \caption{Table of frequently visited parameters.}
    \label{tab:parameters}
\end{table}

In this paper, as representative example of computational jobs, we consider the distributed multiplication of two large matrices using PolyDot scheme \cite{Dutta2020}. This coding scheme will be reviewed in more detail in Appendix~\ref{subsec:preliminaries}. The reason we choose Polydot scheme for distributed matrix multiplication in this paper is the flexibility it offers to balance the computational complexity (minimum required workers) and computational traffic per worker. This flexibility can be utilized later, via some hyper parameters, in designing the bigger solution in order to obtain the desired trade-off among the system performance metrics. We note that the proposed stream computing solution is compatible with any distributed computational job in form of (\ref{equ:distribyted-job}), e.g., Polynomial codes, Matdots codes, $d$-dimensional product codes \cite{BaharavISIT2018}, Short-Dot codes {\cite{dutta2019short}}, and many others. Table~\ref{tab:parameters} shows a list of frequently visited parameters throughout this paper along with their short definitions.

%%%%%%%%%%%%%%%%%%%%%%%%%%%%%%%%%%%%%%%%%%%%%%%%%%%%%%%%%%
\section{Scheduling and Failure Control Solution for Stream Coded Computing}\label{sec:sol}
%%%%%%%%%%%%%%%%%%%%%%%%%%%%%%%%%%%%%%%%%%%%%%%%%%%%%%%%%%

The goal of this section is proposing a systematic framework to identify the parameters of the system, i.e., selecting the set of workers and splitting the computational load among them, such that a stable computational solution with high throughput (high utilization of the acquired workers) and low in-order execution delay (in-order delivery delay of the jobs at the fusion node) is obtained. Our proposed solution efficiently combines the distributed computation idea with scheduling mechanism to improve the throughput and lower the in-order execution delay.

When a job arrives at the master node, the master node selects an appropriate set of workers and distribute the computational load among them. The first $K$ tasks related to a job (the minimum number of tasks for completion of a job) are called \textit{critical tasks}\footnote{Since in our model, we may assign multiple tasks per job to a worker, the condition $K\leq P$ is not necessary.}. The subsequent tasks related to a job are called \textit{supportive tasks}, which compensate for the system failures. Later in Section~\ref{sec:FB}, we introduce a variant of the solution that tracks the real-time realization of errors and delays during one job service time to further improve the performance by assigning feed-back based supportive tasks.

\begin{definition}[Task generation rate]\label{def:TaskGeneration}
Task generation rate of the system is defined as
$K\Omega$, which denotes the number tasks generated per job. The parameter $\Omega\geq1$ is called redundancy ratio.
\end{definition}
\noindent The $(\Omega-1)K$ supportive tasks per job compensate for the system stragglers, i.e., the workers that take much longer than expected to finish an assigned task. Thus, $\Omega$ provides a trade-off between the system throughput and delay such that larger $\Omega$ results in more redundancy (lower throughput) and also lower delay.

\begin{assum}\label{assum:no-bottleneck}
For each worker with index $p$, $p\in\{1,\dots P\}$, the incoming traffic rate $c_p/I_\text{in}$, the outgoing traffic rate $c_p/I_\text{out}$, the encoding rate of the master node, and the decoding rate of the fusion node are all equal to or larger than the average job service rate $1/E[T_p]$. Thus, the encoding, decoding, and communication are not sources of a  bottleneck for the workers.
\end{assum}
\noindent We note that if Assumption~\ref{assum:no-bottleneck} does not hold an extra delay is introduced to the system that needs to be incorporated in the analysis. In case the communication time from the master node to each worker is negligible, i.e., $c_p\rightarrow\infty$, the workers can use the master node buffer.

\begin{definition}[Valid Worker]\label{def:valid_coding} A worker is valid for a distributed coded computation if its parameters satisfy Assumption~\ref{assum:no-bottleneck}. In other words,
\begin{equation*}\label{equ:valid_coding}
\min\paren{\frac{1}{E[T_\text{enc}]},\frac{c_p}{I_{\text{in}}},\frac{c_p}{I_{\text{out}}},\frac{1}{E[T_\text{dec}]}}\geq \frac{1}{E[T_p]}.
\end{equation*}
Here, $E[T_\mathrm{enc}]$ is the average encoding time (for generation of $K\Omega$) at the master node, and $E[T_\mathrm{dec}]$ is the average decoding time per job at the fusion node (given $K$ successful task results).
\end{definition}

\begin{assum}\label{assu:stability1}
For a stable system,
\begin{align*}
    \lambda \leq \sum_{p=1}^P 1/E[T_p].
\end{align*}
\end{assum}
\noindent Moreover, since $\phi_p$ portion of a job is assigned to the $p$-th worker, $p\in\{1,\dots,P\}$, on average one job is assigned to the $p$-th worker during $\lambda\phi_p$ time steps. For stability of each worker, this time must be equal or greater then the average time it takes for the $p$-th worker to perform one job, i.e., $E[T_p]$, resulting in the following stronger stability condition.
\begin{assum}\label{assu:stability2}
For stability at the workers,
\begin{align*}
    \phi_p \leq\rcomp,\quad\forall  p\in\{1,\dots,P\}.
\end{align*}
\end{assum}
\noindent Since $\sum_{p=1}^{P}\phi_p=1$ in the proposed parallel setting, Assumption~\ref{assu:stability2} guarantees Assumption~\ref{assu:stability1}, and it needs to hold for any stable solution to avoid  queues' overflows.

For a choice of hyper parameters for the distributed coded computation scheme, e.g., $s$ and $t$ for the PolyDot scheme, a set of valid codes that satisfy Assumption~\ref{assu:stability1} are selected. The optimal load split among the selected workers, i.e., $\{\phi_1,\dots,\phi_P\}$, is determined to minimize the average job execution time under Assumption~\ref{assu:stability2}. The best choice for hyper parameters of the computation scheme along with the corresponding set of workers and the optimal load split among them are recorded for utilization. The procedure is repeated periodically or whenever changes occur in the system.

%%%%%%%%%%%%%%%%%%%%%%%%%%%%%%%%%%%%%%%%%%%%%%%%%%%%%%%%%%
\subsection{Optimized Split of Computational Load}

Since the problem setting is heterogeneous and workers have various computational powers and transmission rates, the computational load at the master node needs to be split appropriately among the workers. The parameter $\phi_p\in[0,1]$ denotes the fraction of computational load on the $p$-th worker, such that $\sum_{p=1}^{P} \phi_p=1$. Considering the M/G/1 queuing model of the workers, the job arrival rate is $\lambda\phi_p$, the first and second moments of the job service time are $E[T_p]$ and $E[T_p^2]$, respectively, and the work load at the queue at queue of the $p$-th worker is $\rho_p ={\lambda\phi_p}{E[T_p]}$. Hence the average time it takes for the $p$-thworker to respond to a job, including both waiting time in the queue and the processing time, is given by Pollaczek-Khinchin formula \cite[Chapter 5]{gallager2013stochastic},
\begin{equation*}
    D_{\text{comp},p}=\frac{\lambda\phi_p {E}[T_p^2]}{2 (1-\rho_p)}+E[T_p]=0.5\frac{E[T_p^2]}{E[T_p]}\frac{\phi_p}{\frac{1}{\lambda E[T_p]}-\phi_p}+E[T_p]=\frac{1}{\lambda}\paren{ \frac{a_p\phi_p}{\rcomp-\phi_p}+\frac{1}{\rcomp}}.
\end{equation*}
We remind that $\rcomp={1}/{\lambda E[T_p]}$ and $a_p\triangleq0.5\lambda{E[T_p^2]}/{E[T_p]}$.

\begin{definition}[Average job computation time]
The average job computation time is defined as the average response time to a job averaged over all workers and normalized by $P$ (because of the parallel setting),
\begin{equation*}
    D_{\text{comp}}=\frac{1}{P}\sum_{p=1}^P \frac{1}{\lambda}\paren{ \frac{a_p\phi_p^2}{\rcomp-\phi_p}+\frac{\phi_p}{\rcomp}}.
\end{equation*}
\end{definition}

We remind that $\lambda\rcomm=c_p/(I_\text{in}+I_\text{out})$ is the average job transmission rate for the $p$-th worker, $p\in\{1,\dots,P\}$, including both communication from master node to the worker and from the worker to the fusion node.

\begin{definition}[Average job communication time]
It is defined as the average job transmission time averaged over all workers and normalized by $P$ (because of the parallel setting),
\begin{equation*}
    D_{ \text{comm}} = \frac{1}{P}\sum_{p=1}^P\frac{\phi_p}{\lambda\rcomm},
\end{equation*}
\end{definition}

We remind that  $E[T_\text{enc}]$ and $E[T_\text{dec}]$ denote the average encoding time per job, i.e., generating $K\Omega$ tasks at the master node,  and decoding time per job, i.e., decoding the final job result given receiving $K$ successful task results, respectively.
\begin{definition}[Average job execution time]
The average job execution time is given by
\begin{align}\label{equ:Dexe}
D_\text{exe}&=D_{ \text{comp}}+D_{ \text{comm}}+E[T_\text{enc}]+E[T_\text{enc}].
\end{align}
\end{definition}
As the workers are heterogeneous, the way the computational load is split among them must be optimized to minimize the average job execution time\footnote{We do not consider the time-out for the task split. By an appropriate choice of time-out to reduce the failure probability, no time-out assumption is a fair assumption.}. The problem of optimal load split is given below:
\begin{opt}[Optimal Split]\label{opt:rate-split1}
\begin{align*}
    \boldsymbol{\phi}^\star =&
    \argmin{\boldsymbol{\phi}} \sum_{p=1}^P \paren{ \frac{a_p\phi_p^2}{\rcomp-\phi_p}+\left(\frac{1}{\rcomp}+\frac{1}{\rcomm}\right)\phi_p},\\
    \st\quad& \phil\leq \phi_p \leq \rcomp \quad \forall p\in\{1,\hdots,P\}, \nonumber\text{ and } \sum_{p=1}^P\phi_p=1,
\end{align*}
where the objective is minimizing the average job execution time. The constraints are: 1) $\boldsymbol{\phi}$ must be a split (i.e. be positive and sum up to one); 2) all workers must be utilized at least to a certain amount $\phil$; 3) each worker must have a stable queue.
\end{opt}

\begin{theorem}[Optimal~Split]
\label{theorem:optimal_split_comm}
The solution of Optimization Problem~\ref{opt:rate-split1} is given by,
\begin{equation}\label{equ:opt_task_split}
    \phi_p^*=\begin{cases}
     \max\left\{\rcomp \left(1-\sqrt{\frac{a_p}{a_p+\eta-\frac{1}{\rcomp}-\frac{1}{\rcomm}}}\right),\phil\right\},&\frac{1}{\rcomp}+\frac{1}{\rcomm}-a_p<\eta,\\
     \phil,&\text{otherwise.}
    \end{cases}
\end{equation}
where $\eta$ is set such that $\sum_{p=1}^P \phi_p = 1$.
\end{theorem}
\noindent{\em Proof:} The proof is given in Appendix~\ref{theorem:optimal_split_comm_proof}.

Theorem~\ref{theorem:optimal_split_comm} is an analytical solution for Optimization Problem~\ref{opt:rate-split1}. We note that $\phi_p$ is non-decreasing in terms of $\eta$, and thus one can do a binary search\cite{bentley1975multidimensional,nowak2008generalized} to identify $\eta$ with a desired precision such that $\sum_{p=1}^{P}\phi_p=1$.
\begin{example}
If workers are modeled with M/M/1 queue, i.e., $a_p=\nicefrac{1}{\rcomp}$, the optimal split is,
    \begin{equation*}
    \phi_p=\begin{cases}
     \max\left\{\rcomp\left(1-\sqrt{\frac{\frac{1}{\rcomp}}{\eta-\frac{1}{\rcomm}}}\right),\phil\right\},&\frac{1}{\rcomm}<\eta,\\
     \phil,&\text{otherwise,}
    \end{cases}
    \end{equation*}
    where $\eta$ is set such that $\sum_{p=1}^P \phi_p = 1$.
\end{example}

When the transmission delay is negligible, namely, $c_p\rightarrow\infty$ and consequently $\rcomm\rightarrow\infty$, the optimal load split can obtained by the following according to Corollary~\ref{colol:optimal_split_comm}, which is the solution for a special case of Optimization Problem~\ref{opt:rate-split1}.
\begin{corollary}\label{colol:optimal_split_comm}
When the communication delay is negligent, i.e., $c_p\rightarrow\infty$ and $\rcomm\rightarrow\infty$, the optimal split is,
    \begin{equation*}
    \phi_p=\begin{cases}
     \max\left\{\rcomp\left(1-\sqrt{\frac{a_p}{a_p+\eta-\frac{1}{\rcomp}}}\right),\phil\right\},&\frac{1}{\rcomp}-a_p<\eta,\\
     \phil,&\text{otherwise.}
    \end{cases}
    \end{equation*}
    where $\eta$ is set such that $\sum_{p=1}^P \phi_p = 1$.
\end{corollary}

\begin{theorem}[Optimal~Split for $\rcomm\rightarrow\infty$ and exponential service time]\label{theorem:optimal_split_no_comm}
Let assume the workers are sorted such that $ r_1^\text{comp}\geq r_2^\text{comp}\geq\dots\geq r_P^\text{comp}$. For M/M/1 queuing model and when there is no communication delay, the optimal load split is obtained via,
    \begin{equation*}
    \phi_p=\begin{cases}
    \rcomp\left(1-\sqrt{\frac{1}{\eta\rcomp}}\right),&p\leq p^*\\
    \phil&\text{otherwise.}
    \end{cases}
    \end{equation*}
    where $p^*\in\{1,\dots,P\}$ is selected such that $\frac{r^\text{comp}_{p^*}}{(r^\text{comp}_{p^*}-\phil)^2}<\eta(p^*)\leq\frac{r^\text{comp}_{p^*+1}}{(r^\text{comp}_{p^*+1}-\phil)^2}$, $\eta=\eta(p^*)$, and the function $\eta(p^+)$ is defined as follows,
    \begin{equation*}
        \eta(p^+)\triangleq\paren{\frac{\sum_{p=1}^{p^+}\sqrt{\rcomp}}{\sum_{p=1}^{p^+}\rcomp+\sum_{p=p^++1}^{P}\phil-1}}^2.
    \end{equation*}
\end{theorem}
\noindent{\em Proof:} The proof is given in Appendix~\ref{theorem:optimal_split_no_comm_proof}.

%%%%%%%%%%%%%%%%%%%%%%%%%%%%%%%%%%%%%%%%%%%%%%%%%%%%%%%%%%
\subsection{Algorithm for Stream Distributed Coded Computing}

Given a set of code parameters, we first propose in our solution how to select a set of workers that provide sufficient computational resources for the setting, i.e., they are valid and their cumulative task service rate holds Assumption~\ref{assu:stability1} with a margin $\Theta$ for the $(\Omega-1)K$ supportive tasks per job (See Definition~\ref{def:TaskGeneration}).
\begin{opt}\label{opt:choosing_workers} Let ${\overline{\mathcal{P}}}$ be a large set of available workers that are sorted in a decreasing order according to their computational power. The subsets of workers $\mathcal{P}\subset{\overline{\mathcal{P}}}$ that are acquired for the distributed computation are the first valid $P$ workers in ${\overline{\mathcal{P}}}$ such that,
\begin{equation*}
    \sum_{p=1}^{P}\rcomp \geq 1+\Theta \quad\text{s.t.}\quad \Theta\geq (\Omega-1).
\end{equation*}
\end{opt}
\noindent We remind that $\rcomp= \nicefrac{1}{{\lambda E[T_p]}}$. We assume the set of given workers provide sufficient combined computational power for the job arrival rate $\lambda$, and the above optimization problem is feasible. We remind that $E[T_p]$, $p\in\{1,\dots,P\}$, depends on the code parameters. Denoting by $\mathcal{C}$ the set of options for the code parameters, each element of $\mathcal{C}$  corresponds to a choice for parameters in $\{K,I_\text{in},I_\text{out},C,E[T_\text{enc}],E[T_\text{dec}]\}$.

The overall solution can be described as follows: For any member of the set $\mathcal{C}$ (the choices for the code parameters), a subset of workers is selected according to Optimization~Problem~\ref{opt:choosing_workers}. Then, the optimal load split $\{\phi_1,\dots,\phi_P\}$ are identified using Optimization problem~\ref{opt:rate-split1}, and the corresponding average job execution time $D_\text{exe}$ is recorded. The best valid member of the set $\mathcal{C}$ (resulting in the minimum job execution) along with the corresponding optimal load split are recorded for utilization. The procedure is shown in Algorithm~\ref{alg1}, and run in initial stage at the master node or when changes occur in the system, e.g., according to the feedback, due to a worker unavailability, etc.

\begin{algorithm}
\SetAlgoLined
\KwResult{ $\{K,I_\text{in},I_\text{out},C,E[T_\text{enc}],E[T_\text{dec}]\}^*$, $\mathcal{P}^*$, and $\{\phi_1,\dots,\phi_{P}\}^*$.}
\textbf{Initialization:} $D_\text{exe}^*=\infty$.\\
\For{$\{K,I_\text{in},I_\text{out},C,E[T_\text{enc}],E[T_\text{dec}]\}\in \mathcal{C}$}{\text{ }\\
    Choose $\mathcal{P}$ using Optimization Problem~\ref{opt:choosing_workers}.\\
    Find  $\{\phi_0,\dots,\phi_{P-1}\}$ using Optimization Problem~\ref{opt:rate-split1}.\\
    Compute $D_\text{exe}$ using (\ref{equ:Dexe}).\\
    \If{$D_\text{exe}<D_\text{exe}^*$}{\text{ }\\
        $\{K,I_\text{in},I_\text{out},C,E[T_\text{enc}],E[T_\text{dec}]\}^*:=\{K,I_\text{in},I_\text{out},C,E[T_\text{enc}],E[T_\text{dec}]\}$\\
        $\{\phi_1,\dots,\phi_{P}\}^*:=\{\phi_1,\dots,\phi_{P}\}$\\
        $\mathcal{P}^*=\mathcal{P}$
    }
 }
 \caption{Algorithm for Stream Distributed Coded Computation}
 \label{alg1}
\end{algorithm}

\begin{example}[Matrix Multiplication using PolyDot]\label{exam:PolyDot}
For the PolyDot codes scheme, with any positive integers $s$ and $t$ to split the matrices as defined in Appendix~\ref{subsec:preliminaries}, the set of codes can be identified as follows:
\begin{equation*}
    \mathcal{C}=\{(s,t)|1\leq s , t \leq m \text{ and }st=m\}.
\end{equation*}
The other code parameters are given in terms of $s$ and $t$, as follows:
\begin{itemize}
    \item The number of critical tasks is $K=t^2(2s-1)$.
    \item The master node transmits $I_\text{in}=2K\Omega N^2/{st}$ symbols to designated workers per job, and the workers transmit $I_\text{out}=K\Omega N^2/t^2$ symbols to the fusion node\footnote{In case of purging, $I_\text{out}=KN^2/t^2$.}.
    \item The computational complexity of each job is $C=K\Omega{N^3}/{st^2}$.
    \item The average encoding time is
    $E[T_\text{enc}]=K\Omega N^2/\mu_\text{enc}$, and
    the average decoding time is $E[T_ \text{dec}]=(N^2K+K^3)/\mu_\text{dec}$. Here $\mu_\text{enc}$ and $\mu_\text{dec}$ denote the computational power of the master node and fusion node to perform computational powers\footnote{Note that these computational complexities correspond to specific encoding and decoding algorithms. Information regarding other encoding and decoding algorithms can be found in \cite{Dutta2020}.}.
\end{itemize}
\end{example}

%%%%%%%%%%%%%%%%%%%%%%%%%%%%%%%%%%%%%%%%%%%%%%%%%%%%%%%%%%
\section{Feedback-Based Tracking}\label{sec:feedback}
%%%%%%%%%%%%%%%%%%%%%%%%%%%%%%%%%%%%%%%%%%%%%%%%%%%%%%%%%%
In this section, we propose two feedback-based algorithms to track the state of workers and adjust the system parameters accordingly. The first algorithm given in Section~\ref{sec:worker_estimation} is for adaptive estimation of workers' statistical features. The second algorithm given in Section~\ref{sec:FB} is for posteriori job reinforcement at the workers.

%%%%%%%%%%%%%%%%%%%%%%%%%%%%%%%%%%%%%%%%%%%%%%%%%%%%%%%%%%
\subsection{Adaptive Estimation of Workers' Computational Statistical Features}\label{sec:worker_estimation}
One key feature of our solution is the adaptive feature of the load assignments based on the time-varying average service rate of the workers. To this end, the master node requires to have access to the system parameters $E[T_p]$ and $E[T_p^2]$, $p\in\{1,\dots,P\}$. These parameters are provided by the workers' declaration but are adjusted during the execution according to system realizations. In this subsection, we propose several efficient estimators for this purpose, and the right choice depends on the desired trade-offs among accuracy, information availability, and computational limits. Let $U_p$ be
the task service time of the $p$-th worker, $p\in\{1,\dots,P\}$, which is the time it takes for the $p$-th worker to finish a task. The master node, through feedback, has access to realizations of $U_p$ over the time, and intends to adjust the estimations of $E[T_p]$ and $E[T_p^2]$ accordingly.

Let parameters $E_p(l)$ and $S_p(l)$ be estimations of first and second moments of $U_p$ at time $l$. We identify $E_p(l)$ and $S_p(l)$ according to realizations of $U_p$ up to and including time $l$, i.e., $\{u_p(1),\dots,u_P(l)\}$,
\begin{align*}
    E_p(t)&:=(1-\alpha)E_p(l-1)+\alpha u_p(l),\\
    S_p(t)&:=(1-\beta)S_p(l-1)+\beta u_p^2(l).
\end{align*}
These estimation rules are based on the considered problem setting where the first moment and second moment of the time it takes for the $p$-th worker to finish $\phi_p$ portion of a job is $\phi_p E[T_p]$ and $\phi_p^2 E[T_p^2]$, respectively, and one can consider a different appropriate estimation rules for a different problem setting. Here, $\alpha$ and $\beta$ are forgetting factors that determine the importance of old realizations in estimations. If $\alpha=\beta=1/l$, then all previous samples have equivalent importance in estimations.

Alternatively, the parameters $E_p(l)$ and $S_p(l)$ can be updated each $\Lambda$ time steps according to the following rules,
\begin{align*}
    E_p(l)&:=(1-\alpha)^\Lambda E_p(l-\Lambda)+\sum_{i=1}^{\Lambda}(1-\alpha)^{\Lambda-i}\alpha^{i} u_p(l-i+1),\\
    S_p(t)&:=(1-\beta)^\Lambda S_p(l-\Lambda)+\sum_{i=1}^{\Lambda}(1-\beta)^{\Lambda-i}\beta^{i} u_p^2(l-i+1).
\end{align*}
Finally, considering each job consists of $K\Omega$ tasks with no purging, we adjust estimations of $E[T_p]$ and $E[T_p^2]$ at time step $l$ as follows,
\begin{equation}
    \begin{split}
        E[T_p]= \frac{E_P(l)}{K\Omega},\quad\text{ and }\quad E[T_p^2]\simeq \frac{S_P(l)}{K^2\Omega^2}.
    \end{split}
\end{equation}
The approximation comes from the fact that, in practice, the tasks are performed one by one at the workers. Thus, the time it takes for the $p$-th worker to perform one job is the time it takes for the $p$-th worker to perform $\phi_p K\Omega$ sequential tasks which is summation of several random variables that may have different distribution than the distribution of one of these random variables. According to our empirical results and analytical analysis, these approximations are realistic especially when $K\Omega\phil\gg1$, $p\in\{1,\dots,P\}$.
\begin{example}
Consider $U_p$, $p\in\{1,\dots,P\}$, which is the time it takes for the $p$-th worker to finish a task, has an exponential distribution with parameter $K\Omega\tilde{\mu}_p/C$. Then, $T_p(\phi_p)$, $p\in\{1,\dots,P\}$, which is the time it takes for the $p$-th worker to finish $\phi_p$ portion of a job or equivalently $K\Omega\phi_p$ tasks, i.e., has a Gamma distribution with shape $K\Omega\phi_p$ and scale $\nicefrac{C}{K\Omega\tilde{\mu}_p}$. In this setting, we have
\begin{equation*}
    E[T_p(\phi_p)]=\phi_p\frac{C}{\tilde{\mu}_p}=\phi_p E[T_p],\quad\text{and}\quad E[T_p(\phi_p)^2]=\phi_p\left(\phi_p+\frac{1}{K\Omega}\right)\frac{C^2}{\tilde{\mu}_p^2}\simeq \phi_p^2E[T_p^2],
\end{equation*}
The above approximation is valid if $\phi_p\gg 1/K\Omega$, $p\in\{1,\dots,P\}$. Since, we have $\phi_p\geq\phil$, then $K\Omega\phil\gg1$ ensures the approximations are valid.
\label{expl:gamma}
\end{example}

Moreover, the considered $T_p$ as above estimation is an upper bound since the fusion node only requires $K$ task results to obtain the final job result. In case, the master node remove the tasks related to a resolved job from the queue of the master node (purging), each job consists of less then $K\Omega$ tasks. In this case, one can incorporate the purging probability into the estimation and refine the analysis.

%%%%%%%%%%%%%%%%%%%%%%%%%%%%%%%%%%%%%%%%%%%%%%%%%%%%%%%%%%
\subsection{Feedback-Based Algorithm for Posteriori Job Reinforcement}\label{sec:FB}

In the distributed computation scheme considered, the master node tracks the service rate of workers via the feedback from the workers or fusion node. The optimal task split optimization enables us to optimally split the $K\Omega$ tasks amongst the workers at the time the job arrives. Here, we propose an efficient posteriori adaptive algorithm to adjust the load split according to the workers' current service rate during the job's computation period. The adaptive algorithm proposed can minimize the in-order delivery delay of jobs while maximizing the utilization of the workers.

Let $\mathcal{A}_j(l)$ be a counter at the master that records the number of successful task results related to the $j$-th job up to and including time step $l$, obtained using the feedback information. We consider the case that $\mathcal{A}_j(l)<K$, and thus the job is not resolved at the fusion node. We denote as $l_j$ the time that $j$-th job arrives at the master node and is distributed among the workers, and $l>l_j$ the time that a feedback related to $j$-th job is received by the master node. Besides, $E_l[T_p]$ is the estimation of the job service time for the $p$-th worker at time step $l$. We also denote with $\mathcal{P}(l,j)$ the set of workers that are working on $j$-th job at time step $l$\footnote{If $l>l_{j+1}$, these set of workers are the stragglers for the $j$-th job}. We define $\mu_E^j(l)$ and $\mu_R^j(l)$ as the expected computation rate and realized computation rate of $j$-th job for the workers in $\mathcal{P}(l,j)$, and they are obtained as follows
\begin{equation*}
    \mu_E^j(l)\triangleq\frac{1}{|\mathcal{P}(l,j)|}\sum_{p\in\mathcal{P}(l,j)}\frac{\phi_p}{E_{l_j}[T_p]},\quad\mu_R^j(l)\triangleq\frac{1}{|\mathcal{P}(l,j)|}\sum_{p\in\mathcal{P}(l,j)}\frac{\phi_p}{E_l[T_p]}.
\end{equation*}
Here $\{\phi_1,\dots,\phi_p\}$ are based on the optimal task split upon arrival of the $j$-th job.\\

\begin{algorithm}
\SetAlgoLined
\textbf{Initialization:} $\mu_{FB}^{j}=0$ for each received job.\\
At each time step $l$:\\
 \For{$j\in\{1,\dots,\mathcal{N}(l)\}$}{\text{ }\\
    \If{$\mathcal{A}(j)<K$\vspace{0.2cm}}
    {
    $\mu_{\Delta(j)}=\mu_{E}^{j}(l)+\mu_{FB}^{j}-\mu_{R}^{j}(l)$\vspace{0.2cm}\\
    \If{$\mu_{\Delta(j)} \geq \theta$\vspace{0.3cm}}
    {
        Identify the subset of the workers, $\mathcal{P}^{c}$ that work on jobs with indices $>j$.\vspace{0.2cm}\\
        { }$\mathcal{P}_j=\argmax{\mathcal{P}_s\subset\mathcal{P}^{c}}\sum_{p\in\mathcal{P}_s}\frac{1}{E_l[T_p]},\quad\text{s.t.}\quad \mu_{N}^{j} = \sum_{p\in\mathcal{P}_s} \frac{1}{E_l[T_p]}\leq \mu_{\Delta(j)}$\vspace{0.2cm}\\
        $M_j\triangleq\max\{ \left(K(\Omega-1)-\mathcal{A}_j(l)\right)\mu_{N}^{j},0\}$\vspace{0.2cm}\\
        The master node assigns $\left\lceil\frac{M_j}{\mu_{N}^{j}E_l[T_p]}\right\rceil$ tasks related to the $j$-th job to the $p$-th worker in $\mathcal{P}_j$ .\vspace{0.2cm}\\
        $\mu_{FB}^j:=\mu_{FB}^j+\mu_{N}^{j}.$\vspace{0.2cm}\\
    }
    }
 }
 \caption{Feedback-Based Algorithm for Posteriori Job Reinforcement}
 \label{alg_adp}
\end{algorithm}
\vspace{0.5cm}

Then, $\mu_R^j(l)-\mu_E^j(l_j)$ quantifies the missing service rate due to the stragglers related to the $j$-th job up to and including time step $l$. Upon arrival of any feedback related to the $j$-th job at time step $l$, $\mu_{\Delta(j)}$ is computed according to the following rule,
\begin{equation*}
    \mu_{\Delta(j)}=\mu_{E}^{j}(l)+\mu_{FB}^{j}-\mu_{R}^{j}(l),
\end{equation*}
where the parameter $\mu_{FB}^{j}$ indicates the adjusted computation rate that have been considered so far for the $j$-th job according to feedback. Then, if $\mu_{\Delta(j)}>\theta$, (here $\theta$ is a pre-defined threshold), the master node attempts to assign appropriate number of reinforcing tasks related to this job to an appropriately selected subset of workers $\mathcal{P}_j$ at beginning of their queue.

Let $\mathcal{P}^c$ be a subset of workers that have started to work on subsequent jobs\footnote{It can also include the workers that have worked faster than expected, i.e., $E[T_p]<E_j[T_p]$.}, i.e., $$\mathcal{P}_j^c=\mathcal{P}\setminus\mathcal{P}(l,j)\setminus\mathcal{P}(l,j-1)\setminus\dots\setminus\mathcal{P}(l,1).$$
The appropriately selected subset of workers $\mathcal{P}_j$ to reinforce the $j$-th job is obtained using the following optimization problem,
\begin{equation*}
    \mathcal{P}_j=\argmax{\mathcal{P}_s\subset\mathcal{P}^{c}}\sum_{p\in\mathcal{P}_s}\frac{1}{E_l[T_p]},\quad\text{s.t.}\quad \mu_{N}^{j} = \sum_{p\in\mathcal{P}_s} \frac{1}{E_l[T_p]}\leq \mu_{\Delta(j)}
\end{equation*}
Then, the parameter $M_j$ is set as follows,
\begin{equation*}
    M_j\triangleq\max\{ \left(K(\Omega-1)-\mathcal{A}_j(l)\right)\mu_{N}^{j},0\}.
\end{equation*}
which shows the number of reinforcing tasks for the $j$-th job. Next, $\mu_{FB}^j$ is updated as follows,
\begin{equation*}
    \mu_{FB}^j:=\mu_{FB}^j+\mu_{N}^{j}.
\end{equation*}
The adaptive posteriori procedure is given as Algorithm~\ref{alg_adp}. {Note that in solutions considered in the literature (e.g. in \cite{Dutta2020,mallick2019rateless,dutta2019short}) and also here in previous sections without the adaptive algorithm, if some workers stop/halt during the execution, i.e., their service rate approaches zero for some period of time, the system will be prone to some failure probability. Using the proposed adaptive algorithm, we will compensate for this missing service rate due to the workers that stop to work by other workers.}

%%%%%%%%%%%%%%%%%%%%%%%%%%%%%%%%%%%%%%%%%%%%%%%%%%%%%%%%%%
\section{Numerical Results and Discussion}\label{sec:sim}
%%%%%%%%%%%%%%%%%%%%%%%%%%%%%%%%%%%%%%%%%%%%%%%%%%%%%%%%%%
In this section, we evaluate the performance of the proposed joint coding and scheduling framework for stream distributed coded  computation. We first evaluate the method for selecting a subset of the workers from the available ones $\overline{\mathcal{P}}$. Then, we compare the performance of our framework in light of the trade-off between computational load and execution time. The coding scheme we consider here is the PolyDot code (see Appendix~\ref{subsec:preliminaries} and Example~\ref{exam:PolyDot}) with parameters $N=100$, $m=50$. Therefore, the set of all codes can be described as
$\mathcal{C}=\{(s,m/s)|m\text{ is divisible by } s  \}$. For sake of illustration's clarity, we have considered $s$ can be any integer in $\{1,\dots,m\}$, and in practice, one should pick a value of $s$ that is closest to the optimal value and
divides $m$.
%We assume that if a pair $(s,t)$ is selected, one is able to find parameters $(s^{\dagger},t^{\dagger},m^{\dagger})$ close to the original ones such that it forms a valid code\footnote{For instance, if $N=100$, $m=50$, and we would like to have $(s,t)=(26,1.923)$, then a valid code would be $(s^{\dagger},t^{\dagger},m^{\dagger})=(26,2,52)$.}.
In this set of simulations, the task service time of each worker is drawn for an exponential distribution with rate parameter $\Tilde{\mu}_p$ operations per time step. As discussed in Example~\ref{expl:gamma}, this results in a job service time following a Gamma distribution.

%%%%%%%%%%%%%%%%%%%%%%%%%%%%%%%%%%%%%%%%%%%%%%%%%%%%%%%%%%
\subsection{Evaluation of the Initial Choice of workers}
\begin{figure}
  \centering
    \scalebox{0.87}{\hspace{-0.8cm}\input{choose_workers02.tikz}}
  \caption{Exploring the effect of the initial set of workers on the system parameters and performance. (a): The number of valid and partially valid workers. (b): Total computational rate of the selected workers versus the code parameter $s$. (c): The average job execution time versus the code parameter $s$, along with its constituent parts, when the stability conditions hold.}
  \label{fig:initial_choice}
\end{figure}
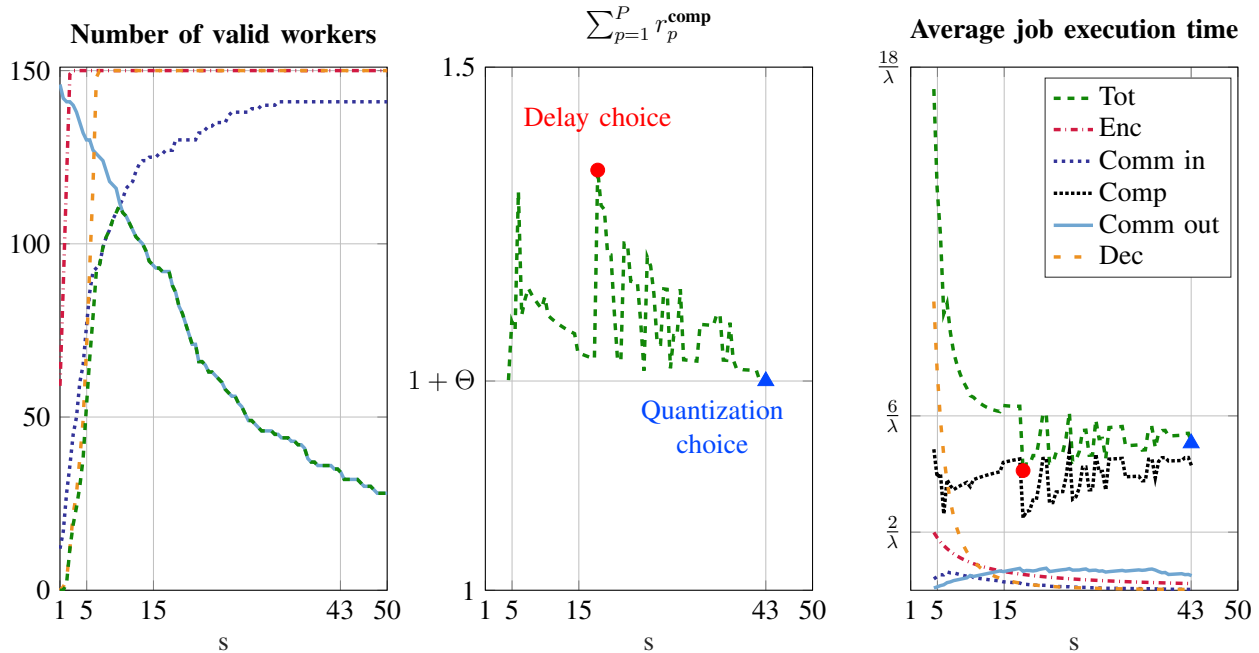
First, we illustrate and evaluate the impact of our method for choosing the workers. In this set of simulations, the rate parameter $\tilde{\mu}_p$ of each worker is uniformly and independently selected in $[0,1000]$ operations per time step, the transmission rate $c_p$ of each worker is uniformly and independently selected in $[0,200]$ symbols per time step, $\mu_\text{enc} = 10000$ operations per time step, and $\mu_\text{dec} = 100000$ operations per time step. Given the system parameters, one can easily find $1/E[T_p]$, $c_p/I_\text{in}$, and $c_p/I_\text{out}$ for each worker in $\overline{\mathcal{P}}$, along with $1/E[T_\text{enc}]$ and $1/E[T_\text{dec}]$. We remind that the set of valid workers is a subset of $\overline{\mathcal{P}}$ such that computation is the bottleneck for each one of them. In other words, the $p$-th worker in $\overline{\mathcal{P}}$ is valid if $1/E[T_p]$ is smaller than $c_p/I_\text{in}$, $c_p/I_\text{out}$, $1/E[T_\text{enc}]$, and $1/E[T_\text{dec}]$, see Definition~\ref{def:valid_coding}.

Fig.~\ref{fig:initial_choice}(a) explores a realization of $\overline{\mathcal{P}}$ that consists of $150$ workers with respect to the code parameter $s$. The red dashed-dotted curve shows the number of workers such that their computation rate is lower than the encoding rate, i.e., $\{p\in\overline{\mathcal{P}}\text{ | }1/E[T_p]<1/E[T_\text{enc}]\}$. Similarly, the dashed dark blue, resp. solid light blue and dashed orange, curve shows the number of workers such that their computation rate is lower than their incoming traffic rate, resp., outgoing traffic rate and the decoding rate. The dashed green curve shows the number of valid workers which stands below all these curves. We observe that for small values of $s$, the decoding defines an upper bound for the number of valid workers, while for large values of $s$, the outgoing traffic rate defines an upper bound.
%Between these two regimes, i.e., i.e. when $5\leq s\leq 10$, incoming traffic rate defines .
The highest number of valid workers is obtained when $s=10$.

Once the valid workers are identified, a subset of them is selected to serve the queue of in-order jobs at the master node, using Optimization~Problem~\ref{opt:choosing_workers}. Fig.~\ref{fig:initial_choice}(b) shows the value of $\sum_{p=1}^{P}\rcomp$, which must be equal or greater than $1+\Theta$, with respect to the code parameter $s$ when $\lambda=10^{-3}$. We exclude those values of $s$ such that the valid workers do not provide sufficient resources to serve the queue of the jobs, i.e. when Optimization~Problem~\ref{opt:choosing_workers} is not feasible. This happens in this example when $s<5$ or $s>43$. Thus, we restrict ourselves to admissible values of $s$ that result in feasible optimizations.

The optimal choice of $s$ is the one that results in the minimum average job execution time.
Fig.~\ref{fig:initial_choice}(c) shows the average job execution time with {dashed green line} for various admissible values of $s$. Moreover, it shows the constituent components of the average job execution time, i.e., encoding time, communication time from the master node to workers, computation time at the workers, communication time from the workers to the fusion node, and decoding time. We observe that the average job execution time with respect to $s$ has several local minima, and selecting any of them can be a viable choice. The fluctuations in the curve is due to the quantization errors that enter into the design in \Cref{opt:choosing_workers}. Hence, in this example, choosing any value of $s$ between $15$ and $43$ is reasonable. Several criteria can be considered for selecting the optimal $s$. First, one can choose $s$ leading to the minimum delay (computation time), as depicted by a red point in the figure. This case will likely correspond to the highest quantization error (i.e., reduced utilization). Another choice is to select $s$ such that the delay remains reasonable while minimizing the quantization error (i.e., maximum utilization), as depicted by a blue triangle in the figure.
%This case could be of interest if the workers must be bought according to their computational power, and when therefore quantization errors translate into an increase of the costs.
Finally, one should also take into account that $s$ must divide $m$ in order to obtain a valid code.
%Other criteria can be designed and we have not pushed this discussion further, keeping it as future work.

%%%%%%%%%%%%%%%%%%%%%%%%%%%%%%%%%%%%%%%%%%%%%%%%%%%%%%%%%%
\subsection{Trade-off between computational load and average job execution time}

\begin{figure}
     \centering
     \begin{subfigure}[b]{0.48\textwidth}
         \centering
         \scalebox{0.72}{\input{Comparisonheterogeneousworkers_withgini.tikz}}
            \caption{Average execution time versus computational load}
            \label{fig:meanexedelay}
     \end{subfigure}
     \hfill
     \begin{subfigure}[b]{0.48\textwidth}
         \centering
         \scalebox{0.72}{\input{Comparisonheterogeneousworkers_with_gini_zoom.tikz}}
     \caption{Average execution time versus computational load}
         \label{fig:meanexedelay_zoom}
     \end{subfigure}
     \hfill
     \begin{subfigure}[b]{0.48\textwidth}
         \centering
         \scalebox{0.72}{\input{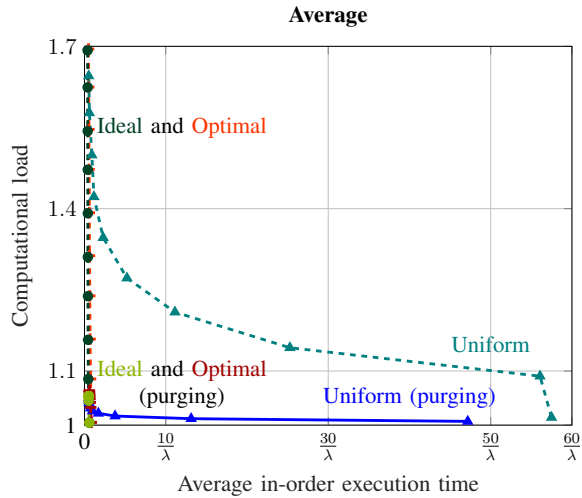}}
     \caption{Realizations of execution time versus computational load}
         \label{fig:real}
     \end{subfigure}
     \hfill
     \begin{subfigure}[b]{0.48\textwidth}
         \centering
         \scalebox{0.72}{\input{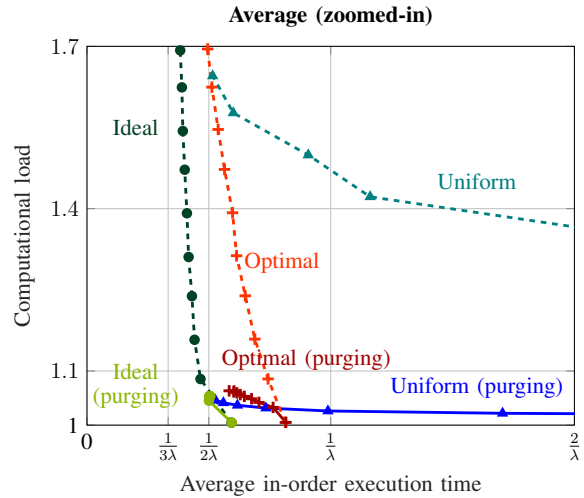}
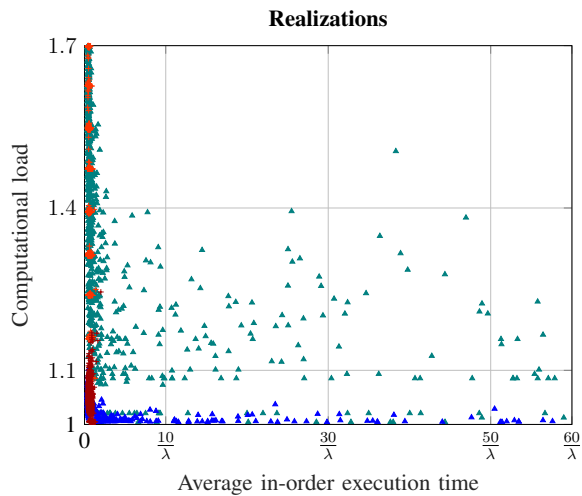
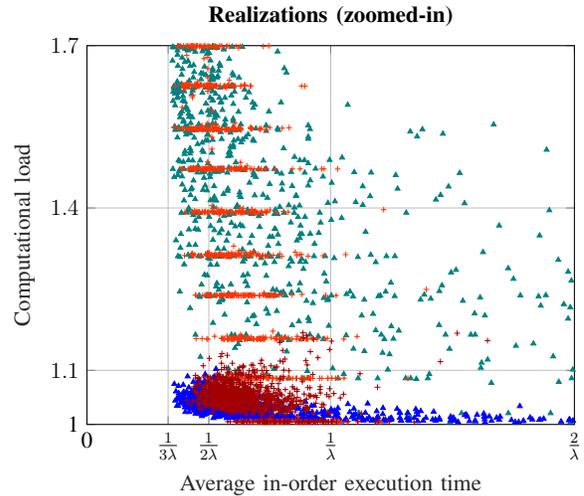}
     \caption{Realizations of execution time versus computational load}
         \label{fig:real_zoom}
     \end{subfigure}
         \caption{Trade-off between computational load and job execution time for a cluster of heterogeneous workers. The right figures are zoomed-in versions of the left ones. The bottom figures highlight $100$ realizations out of the $500$ which have been used to obtain the average execution time for the above ones (For the sake of clarity, the realizations of the ideal case have not been drawn.). The performances have been averaged over $500$ realizations with $200$ jobs each,  $\lambda=10^{-3}$, $m=50$, $N=100$, $s=40$, $t=10$, $\Theta=2$, and the forgetting factor of the mean and variance estimations equal to $0.01$. The parameter $\Omega$ changes in $[1,1.7]$. Finally, we note that $\frac{1}{3\lambda}$ is the inverse of the total service rate. This delay corresponds to the computation delay (i.e. neglecting communication, encoding and decoding) when jobs are served instantaneously (i.e. neglecting the time in the queue).}
    \label{fig:trade_off_delay_comp_load}
\end{figure}

Here, we investigate the performance of our framework with respect to the performances of two extreme baselines: 1) Naive load split where the computational load of a job is distributed uniformly among several workers regardless of their differences. 2) Ideal load split\footnote{ Since, in our work, we assume the communication delay exists, the ideal split works in a  non-causal way such that the next task arrives precisely when the worker becomes idle.} where the master node gets informed when a worker finishes its assigned task and outsources another task to it. In this set of simulations, the rate parameter $\tilde{\mu}_p$ of each worker is uniformly and independently selected in $[0,2500]$ operations per time step, the transmission rate $c_p$ of each worker, is uniformly and independently selected in $[0,1000]$ symbols per time step, $\mu_\text{enc} = 10^5$ operations per time step, and $\mu_\text{dec} = 10^6$ operations per time step \footnote{In this set of simulations, we do not use Optimization~Problem~\ref{opt:choosing_workers} to select the workers to investigate a heterogeneous cluster of workers. Hence, we are facing fast and slow workers.}.

We select $\Theta=2$, i.e., $\sum_{p=1}^P\rcomp=3$, and consequently, the set of workers can serve a queue of jobs with average arrival rate up to $3\lambda$. In fact, $\Theta$ affects the number of feasible solutions for  Optimization Problem~\ref{opt:rate-split1}, and a lower $\Theta$ lowers the impact of the optimal load split. We first remind that $D_\text{exe}$ is the average execution time per job, and $\Omega$ determines the computational load, i.e., the ratio of the total amount of computations (including the redundant computations) to necessary amount of computations. Fig.~\ref{fig:trade_off_delay_comp_load} depicts the computational load versus average job execution time for various interesting scenarios where $\lambda = 10^{-3}$ and the redundancy ratio $\Omega$ varies in interval $[1,1.7]$. We highlight that $\sum_{p=1}^{P}\rcomp$, i.e., the total computational rate of the workers per job, depends on $\Omega$, which determines the number of tasks per job and the selected set of workers to serve $K\Omega$ tasks by job.

We first note that the \textit{uniform split}, i.e., $\phi_p=1/P$ and $p\in\{1,\dots,P\}$, does not necessarily result in a stable solution as confirmed by our simulation results. The \textit{ideal split}, i.e., genie-aided, has a time variant load split, such that the master node assigns a task to a worker as soon as the worker declares the previous task has finished. Lastly, the \textit{optimal split}, is obtained using the proposed framework in this paper. Purging is an optimal feature where the master node asks all workers to drop their tasks related to a job once the decoder is able to decode the job result using the successful tasks up to that time. Fig.~\ref{fig:trade_off_delay_comp_load} shows the trade-off between the average job execution time and the computational load for the various split schemes with and without purging. In case, there is no purging, the computational load scales proportional to $\Omega$ as all tasks are computed\footnote{The uniform solution seems to have a lower computational load than the other methods at high $\Omega$. This is a simulation artefact: as the queues are not stable, the simulation is ended before all unnecessary tasks of the unstable workers are computed. Hence, this translates into a utilization which does not scale exactly as $\Omega$.}. When there is purging, the computational load gets closer to $1$.

As we see in Fig.~\ref{fig:trade_off_delay_comp_load}, the average execution time of the uniform split is much larger than the optimal and ideal splits. This is because a subset of workers are unstable in this naive split. When the computational load increases, the average job execution time for the uniform scheme decreases since the supportive task results of the stable workers can compensate for the unstable workers. More importantly, we observe that the performance of the proposed optimal split is close to the performance of the genie-aided split, although our approach relies on negligible feedback from the fusion node for estimations.
From  Fig.~\ref{fig:trade_off_delay_comp_load}(c) and (d), which show the realizations of $100$ of the $500$ independent runs which have been used to generate Fig.~\ref{fig:trade_off_delay_comp_load}(a) and (b), one can also observe the reliability of our approach with respect to the uniform split method. Indeed, not only the average performance is better with our framework, but also all runs have close performances. This is very different from the uniform case where the choice of the workers greatly impact the performance.

In all cases, when purging is performed, the computational load decreases and approaches to one. We see that when $\Omega=1$, purging has no effect since all tasks are necessary to resolve a job. However, the computational load slightly increases when $\Omega$ increases as more workers are used, and therefore more tasks are likely to arrive at the fusion node before purging is performed. In case of purging and when $\Omega$ is large, the performance of genie-aided, uniform, and optimal split get closer. Interestingly, our optimal solution has a slightly bigger average execution time than the two others for high $\Omega$, with purging. This can be explained by the fact that the optimization problem, based on a queuing model, does not incorporate the purging mechanism into its delay analysis. To reduce this degradation, one can incorporate purging probability into the estimation and refine the analysis.

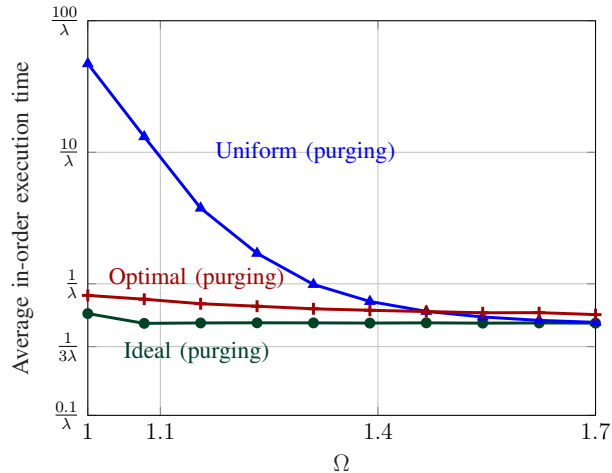
\begin{figure}
\centering
    \scalebox{0.75}{\input{Comparisonheterogeneousworkers_delayOmega.tikz}}
  \caption{Evolution of the in-order execution time with respect to the redundancy ratio $\Omega$. These results are exactly those of Fig.~\ref{fig:trade_off_delay_comp_load}, yet presented in a different way.}
    \label{fig:delay_omega}
\end{figure}

As the computational load remains close to one, in order to better grasp the impact of $\Omega$ on the execution time, Fig.~\ref{fig:delay_omega} shows the evolution of the average in-order execution time as a function of $\Omega$. We again see that our optimal solution remains very close to the ideal non-causal one for all values of $\Omega$. As discussed above, the discrepancy between the three different solutions vanishes at high $\Omega$ since tasks are purged before entering into the regime where the optimization is useful.

%%%%%%%%%%%%%%%%%%%%%%%%%%%%%%%%%%%%%%%%%%%%%%%%%%%%%%%%%%
\section{Conclusions}\label{sim:conc}
%%%%%%%%%%%%%%%%%%%%%%%%%%%%%%%%%%%%%%%%%%%%%%%%%%%%%%%%%%
In this paper, we studied a stochastic heterogeneous setting for the stream distributed coded computation problem. We proposed a systematic framework for the joint scheduling-coding that incorporates the diverse properties of the workers into the design. In particular, an appropriate set of workers from a pool of workers are chosen to provide stability, and the coding parameters and the load split among the selected workers are optimally identified to minimize the average in-order job execution time. Furthermore, a realistic feedback model was introduced to track the state of the workers and the progress of ongoing jobs at the master node.

An interesting future direction is to extend the framework for stream computation of jobs with general arrival model, not necessarily Poisson arrival model. Moreover, proposing
a joint scheduling-coding for stream distributed computation in a network of heterogeneous interconnected workers is another interesting future research direction. Last but not the least, generalizing the proposed framework for stream computation of iterative jobs, often required in machine learning applications, is one of our future goals.

%%%%%%%%%%%%%%%%%%%%%%%%%%%%%%%%%%%%%%%%%%%%%%%%%%%%%%%%%%
\section{Acknowledgments}
%%%%%%%%%%%%%%%%%%%%%%%%%%%%%%%%%%%%%%%%%%%%%%%%%%%%%%%%%%
This work is supported by the European Regional Development Fund (FEDER), through the Regional Operational Programme of Centre (CENTRO 2020) of the Portugal 2020 framework and FCT under the MIT Portugal Program [Project SNOB-5G with Nr. 045929(CENTRO-01-0247-FEDER-045929)]. Guillaume Thiran is a Research Fellow of the F.R.S.-FNRS.

%%%%%%%%%%%%%%%%%%%%%%%%%%%%%%%%%%%%%%%%%%%%%%%%%%%%%%%%%%
\appendices
%%%%%%%%%%%%%%%%%%%%%%%%%%%%%%%%%%%%%%%%%%%%%%%%%%%%%%%%%%

%%%%%%%%%%%%%%%%%%%%%%%%%%%%%%%%%%%%%%%%%%%%%%%%%%%%%%%%%%
\section{Distributed Matrix Multiplication with PolyDot}\label{subsec:preliminaries}
%%%%%%%%%%%%%%%%%%%%%%%%%%%%%%%%%%%%%%%%%%%%%%%%%%%%%%%%%%

Here, we provide background on distributed matrix multiplication with PolyDot scheme \cite{Dutta2020}. The computational job is multiplication of two large matrices $\mathbf{A}$ and $\mathbf{B}$ of size $N\times N$, i.e., $\mathbf{A}\mathbf{B}$, in a distributed fashion. The computational inputs of each job are devised as follows:
\ifdouble
\begin{equation*}
\begin{split}
\mathbf{A}=\left[
\begin{array}{ccc}
\mathbf{A}_{0,0}&\cdots&\mathbf{A}_{0,s-1}\\
\vdots&\ddots&\vdots\\
\mathbf{A}_{t-1,0}&\cdots&\mathbf{A}_{t-1,s-1}
\end{array}\right],\\ \mathbf{B}=\left[
\begin{array}{ccc}
\mathbf{B}_{0,0}&\cdots&\mathbf{B}_{0,t-1}\\
\vdots&\ddots&\vdots\\
\mathbf{B}_{s-1,0}&\cdots&\mathbf{B}_{s-1,t-1}
\end{array}\right],
\end{split}
\end{equation*}
\else
\begin{equation*}
\begin{split}
\mathbf{A}=\left[
\begin{array}{ccc}
\mathbf{A}_{0,0}&\cdots&\mathbf{A}_{0,s-1}\\
\vdots&\ddots&\vdots\\
\mathbf{A}_{t-1,0}&\cdots&\mathbf{A}_{t-1,s-1}
\end{array}\right], \quad  \mathbf{B}=\left[
\begin{array}{ccc}
\mathbf{B}_{0,0}&\cdots&\mathbf{B}_{0,t-1}\\
\vdots&\ddots&\vdots\\
\mathbf{B}_{s-1,0}&\cdots&\mathbf{B}_{s-1,t-1}
\end{array}\right],
\end{split}
\end{equation*}
\fi
$st=m$, $X_{si+j}=[\textbf{A}_{i,j},\textbf{B}_{i,j}]$, and $f(X_1,...,X_m)=\textbf{A}\textbf{B}$.
The master node splits each multiplication job into a set of smaller tasks. The inputs of $r$-th task has the format $[\mathcal{P}_\textbf{A}(x),\mathcal{P}_\textbf{B}(x)]$, evaluated at $x=x_r$, and will be sent to a worker to perform the multiplication over these smaller matrices. The matrices $\mathcal{P}_\textbf{A}(x)$ and $\mathcal{P}_\textbf{B}(x)$ are defined as follows:
\ifdouble
\begin{equation*}
\begin{split}
\mathcal{P}_\textbf{A}(x)&=\sum_{i=0}^{t-1}\sum_{j=0}^{s-1}\textbf{A}_{i,j}x^{tj+i},\\
\mathcal{P}_\textbf{B}(x)&=\sum_{k=0}^{s-1}\sum_{l=0}^{t-1}\textbf{B}_{k,l}x^{t(2s-1)l+t(s-1-k)}.
\end{split}
\end{equation*}
\else
\begin{equation*}
\mathcal{P}_\textbf{A}(x)=\sum_{i=0}^{t-1}\sum_{j=0}^{s-1}\textbf{A}_{i,j}x^{tj+i}, \quad
\mathcal{P}_\textbf{B}(x)=\sum_{k=0}^{s-1}\sum_{l=0}^{t-1}\textbf{B}_{k,l}x^{t(2s-1)l+t(s-1-k)}.
\end{equation*}
\fi

Then, the worker computes $\mathcal{P}_\textbf{C}(x_r)=\mathcal{P}_\textbf{A}(x_r)\mathcal{P}_\textbf{B}(x_r)$. Fusion node uses $K=t^2(2s-1)$ successful task results from the workers to identify the coefficient of term $x^{t(2s-1)(l-1)+t(s-1)+i-1}$ for polynomial $\mathcal{P}_\textbf{C}(x)$, $\forall i,l\in\{0,\dots,t-1\}$, defined below:
\begin{equation*}
\mathcal{P}_\textbf{C}(x)=\sum_{i,j,k,l}\textbf{A}_{i,j}\textbf{B}_{k,l}x^{t(2s-1)l+t(s-1+j-k)+i}.
\end{equation*}
Note this coefficient, which happens when $j-k=0$, is $\mathbf{C}_{i,l}=\sum_{k=0}^{s-1}\textbf{A}_{i,k}\textbf{B}_{k,l}$. The  maximum degree of $p_\textbf{C}(x)$ is $t^2(2s-1)-1$, thus one needs $t^2(2s-1)$ evaluations to identify the polynomial.

The parameters $K$, $I_\text{in}$, $I_\text{out}$, $C$, and the encoding and decoding computational complexities are functions of parameter $s$ and $t$, as we see in Example~\ref{exam:PolyDot}. These parameters influence the average job execution time, and we later optimize these two parameters to achieve a desired trade-off between various performance metrics of our solution.

%%%%%%%%%%%%%%%%%%%%%%%%%%%%%%%%%%%%%%%%%%%%%%%%%%%%%%%%%%
\section{Proof of Theorem~\ref{theorem:optimal_split_comm} - Optimal Split 1}\label{theorem:optimal_split_comm_proof}
%%%%%%%%%%%%%%%%%%%%%%%%%%%%%%%%%%%%%%%%%%%%%%%%%%%%%%%%%%

The proof of Theorem~\ref{theorem:optimal_split_comm} is considered here. We first show that the Optimization Problem~\ref{opt:rate-split1} is convex. Hence, the optimal solution is unique.

\begin{lemma}
\label{lemma:convex}
    The objective function in (\ref{equ:opt_task_split}) is strictly convex on its domain, i.e.,
    $$\boldsymbol{\phi}\in\left\{\{\phi_1,\dots,\phi_P\}|0\leq \phi_p \leq \rcomp\right\}.$$
\end{lemma}
\begin{proof}
The objective function in (\ref{equ:opt_task_split}) is summation of $P$ strictly convex terms and $P$ linear terms, and thus it is strictly convex. We note $f\paren{\phi}=\frac{\alpha\phi^2}{\beta-\phi}$ is strictly convex over $\phil\leq \phi\leq \beta$ when $\alpha>0$ since,
\begin{align*}
    f^{''}\paren{\phi} =\frac{2\alpha}{\beta-\phi}+\frac{2\alpha\phi(2\beta-\phi)}{(\beta-\phi)^3}>0,\qquad \forall\hspace{0.03in} 0\leq \phi\leq \beta.
\end{align*}
\end{proof}

Since the optimization problem has a strictly convex objective function, and the Slater's conditions hold\footnote{Slater's conditions state the strong duality holds if there exists an strictly feasible point in the domain of the concave optimization problem. Thus, once can target the dual problem and use the Lagrangian method\cite[Section 5.2.3]{boyd2004convex}.}, the unique solution can be found using Karush-Kuhn-tucker (KKT) conditions. We first define the Lagrangian function of the problem,
\begin{multline*}
    \mathcal{L}\paren{\boldsymbol{\phi},\boldsymbol{\delta},\boldsymbol{\gamma},\eta} = \sum_{p=1}^P \frac{a_p\phi_p^2}{\rcomp-\phi_p}+\sum_{p=1}^P\left(\frac{1}{\rcomp}+\frac{1}{\rcomm}\right)\phi_p\\+\sum_{p=1}^P\delta_p\paren{\phil-\phi_p}\nonumber+\sum_{p=1}^P\gamma_p\paren{\phi_p-\rcomp}-\eta\paren{\sum_{p=1}^P\phi_p-1}.
\end{multline*}
According to the KKT conditions, the optimal solution of the original optimization problem is a saddle point for the Lagrangian function (the first derivative of the Lagrangian function must be zero at the optimal point),
\begin{equation}\label{equ:Lagrangian1}
\frac{a_p(\rcomp)^2}{\paren{\rcomp-\phi_p}^2}
+\xi_p-\delta_p+\gamma_p-\eta=0,
\end{equation}
where $\xi_p\triangleq \frac{1}{\rcomp}+\frac{1}{\rcomm}-a_p$. It also satisfies the following feasibility conditions:
\begin{align*}
    &\text{Primal feasibility:}\quad\phil\leq \phi_p\leq \rcomp,\quad \sum_{p=1}^P\phi_p=1,\\
    &\text{Dual feasibility:}\quad \delta_p\geq0,\quad \gamma_p\geq 0,\\
    &\text{Complementary slackness:}\quad \delta_p\paren{\phil-\phi_p}=0,\\
    &\hspace{4.55cm}\quad \gamma_p\paren{\phi_p-\rcomp}=0.
\end{align*}

First, we note that $\gamma_p=0$, since $\gamma_p\neq0$ makes $\phi_p=\rcomp$, which cannot be a valid solution for (\ref{equ:Lagrangian1}). We study the cases when $\delta_p>0$ and $\delta_p=0$, separately:

\noindent [\textbf{Case 1:} Workers with $\delta_p>0$] Then, $\phi_p=\phil$ and because of (\ref{equ:Lagrangian1}),
\begin{equation*}
    \delta_p=\delta_p(\eta)=\frac{a_p(\rcomp)^2}{\paren{\rcomp-\phil}^2}+\xi_p-\eta,
\end{equation*}
which must be positive requiring,
\begin{equation*}\label{equ:ineq_eta1}
    \eta<\frac{a_p(\rcomp)^2}{\paren{\rcomp-\phil}^2}+\xi_p.
\end{equation*}
[\textbf{Case 2:} Workers with $\delta_p=0$] Then (\ref{equ:Lagrangian1}) simplifies to,

\begin{align}
    &\frac{a_p(\rcomp)^2}{\paren{\rcomp-\phi_p}^2}+\xi_p-\eta=0,\nonumber\\
    &\paren{\rcomp-\phi_p}^2=\frac{{a_p(\rcomp)^2}}{\eta-\xi_p}\label{equ:phisolve}.
\end{align}
We note that (\ref{equ:phisolve}) is only feasible when $\eta>\xi_p$. In this case, the valid solution for (\ref{equ:phisolve}) is,
\begin{equation*}
    \phi_p=\rcomp\left(1-\sqrt{\frac{a_p}{\eta-\xi_p}}\right).
\end{equation*}
The other solution $\phi_p=\rcomp\left(1+\sqrt{\frac{a_p}{\eta-xi_p}}\right)$ does not satisfy the constraint $\phi_p\leq \rcomp$ and is invalid.
The constraint  $\phi_p\leq\phil$ requires,
\begin{align*}
    &\rcomp\left(1-\sqrt{\frac{a_p}{\eta-\xi_p}}\right)\geq\phil\nonumber\\
    &\frac{a_p}{\eta-\xi_p}\leq\frac{(\rcomp-\phil)^2}{(\rcomp)^2}\nonumber\\
    &\eta\geq\frac{a_p(\rcomp)^2}{\paren{\rcomp-\phil}^2}+\xi_p,
\end{align*}
which is a stronger inequality than $\xi_p>\eta$. Finally, we consider both cases together,
\begin{equation*}
    \phi_p=\begin{cases}
     \rcomp\left(1-\sqrt{\frac{a_p}{\eta-\xi_p}}\right),&\frac{a_p(\rcomp)^2}{\paren{\rcomp-\phil}^2}+\xi_p\leq\eta,\\
     \phil,&\text{otherwise.}
    \end{cases}
\end{equation*}
To incorporate the last remaining constraint, $\eta$ is set such that $\sum_{p=1}^{P}\phi_p=1$. Since the problem has a unique solution, there exists a unique $\eta$ such that above holds.
Conversely, with the above choice, the system of KKT conditions holds with $\gamma_p=0$, $\eta$ such that $\sum_{p=1}^P\phi_p=1$, and $\delta_p=\max\{\delta_p(\eta),0\}$, and this concludes the proof.

%%%%%%%%%%%%%%%%%%%%%%%%%%%%%%%%%%%%%%%%%%%%%%%%%%%%%%%%%%
\section{Proof of Theorem ~\ref{theorem:optimal_split_no_comm}}\label{theorem:optimal_split_no_comm_proof}
%%%%%%%%%%%%%%%%%%%%%%%%%%%%%%%%%%%%%%%%%%%%%%%%%%%%%%%%%%
According to Theorem~\ref{theorem:optimal_split_comm}, and when $a_p=\nicefrac{1}{\rcomp}$ and $\rcomm\rightarrow\infty$, the optimal load split is,
    $$
    \phi_p=\max\left\{\rcomp\left(1-\sqrt{\frac{1}{\eta\rcomp}}\right),\phil\right\}=\begin{cases}
    \rcomp\left(1-\sqrt{\frac{1}{\eta\rcomp}}\right),&\frac{\rcomp}{(\rcomp-\phil)^2}<\eta\\
    \phil,&\text{otherwise.}
    \end{cases}.
    $$
    Since the function $\frac{r}{(r-\phil)^2}$ is strictly decreasing in $r>\phil$ and the workers are sorted such that $ r_1^\text{comp}\geq r_2^\text{comp}\geq\dots\geq r_P^\text{comp}$,
    $$
    \phi_p=\begin{cases}
    \rcomp\left(1-\sqrt{\frac{1}{\eta\rcomp}}\right),&p\leq p^*\\
    \phil&\text{otherwise}.
    \end{cases}
    $$
    Finally, for having $\sum_{p=1}^{P}\phi_p=1$, $\eta=\eta(p^*)$ must be a valid choice, i.e,
    \begin{equation*}
       \frac{r^\text{comp}_{p^*}}{(r^\text{comp}_{p^*}-\phil)^2}<\eta(p^*)\leq \frac{r^\text{comp}_{p^*+1}}{(r^\text{comp}_{p^*+1}-\phil)^2}.
    \end{equation*}
Since the optimization problem has a unique solution, such unique $p^*$ exists.

%%%%%%%%%%%%%%%%%%%%%%%%%%%%%%%%%%%%%%%%%%%%%%%%%%%%%%%%%%
\bibliographystyle{IEEEtran}
\bibliography{references}
%%%%%%%%%%%%%%%%%%%%%%%%%%%%%%%%%%%%%%%%%%%%%%%%%%%%%%%%%%
\end{document}

%% file: choose_workers02.tikz
\definecolor{mycolor1}{rgb}{0.44, 0.65, 0.82}%{0.36, 0.54, 0.66}
%\definecolor{mycolor2}{rgb}{0.00000,1.00000,1.00000}%
\definecolor{colorblue}{rgb}{0.93, 0.57, 0.13}
\definecolor{colorred}{rgb}{0.82, 0.1, 0.26}
	\definecolor{colorgreen}{rgb}{0.2, 0.2, 0.6}%
\definecolor{colorbluedark}{rgb}{0.07, 0.53, 0.03}
% 85 65 13
\begin{tikzpicture}

\begin{axis}[%
width=5cm,
height=8cm,
at={(13cm,0.7cm)},
scale only axis,
scaled y ticks = false,
unbounded coords=jump,
xmin=1,
xmax=50,
xlabel style={font=\color{white!15!black}},
xlabel={s},
xtick={1,5,15,43,50},
xticklabels={{1},{5},{15},{43},{50}},
ymin=0,
ymax=18000,
ytick={2000,6000,18000},
yticklabels={{$\frac{2}{\lambda}$},{$\frac{6}{\lambda}$},{$\frac{18}{\lambda}$}},
axis background/.style={fill=white},
title style={font=\bfseries},
title={Average job execution time},
xmajorgrids,
ymajorgrids,
legend style={legend cell align=left, align=left, draw=white!15!black}
]
\addplot [color=colorbluedark, dashed, line width=1.5pt]
  table[row sep=crcr]{%
1	inf\\
1.4949494949495	inf\\
1.98989898989899	inf\\
2.48484848484848	inf\\
2.97979797979798	inf\\
3.47474747474747	inf\\
3.96969696969697	inf\\
4.46464646464646	17245.3096325186\\
4.95959595959596	13820.0933882322\\
5.45454545454545	12232.1237804693\\
5.94949494949495	9520.5183683088\\
6.44444444444444	9865.86057082185\\
6.93939393939394	8874.65608643342\\
7.43434343434343	8185.13900225091\\
7.92929292929293	7761.47221505689\\
8.42424242424242	7427.97783294141\\
8.91919191919192	7162.83784685103\\
9.41414141414141	6950.29430199263\\
9.90909090909091	6697.58532101046\\
10.4040404040404	6639.63347965027\\
10.8989898989899	6526.34904916113\\
11.3939393939394	6433.91394158759\\
11.8888888888889	6358.44962672473\\
12.3838383838384	6296.91787106799\\
12.8787878787879	6246.91486441723\\
13.3737373737374	6206.52164023898\\
13.8686868686869	6174.19301805203\\
14.3636363636364	6148.67525502349\\
14.8585858585859	6358.80951230982\\
15.3535353535354	6346.27049719178\\
15.8484848484848	6337.77958425636\\
16.3434343434343	6332.76912964381\\
16.8383838383838	6330.76425241839\\
17.3333333333333	6331.36584874347\\
17.8282828282828	4118.25041011359\\
18.3232323232323	4342.94382284545\\
18.8181818181818	4327.80287106406\\
19.3131313131313	4635.49022440018\\
19.8080808080808	4632.80169331096\\
20.3030303030303	5276.92964145097\\
20.7979797979798	6065.90347353611\\
21.2929292929293	6074.72724310377\\
21.7878787878788	4374.24278907005\\
22.2828282828283	4375.18829645873\\
22.7777777777778	4739.0602517526\\
23.2727272727273	4743.04775314002\\
23.7676767676768	4747.73663295699\\
24.2626262626263	5237.10900391904\\
24.7575757575758	6115.17493607261\\
25.2525252525253	4416.13949530224\\
25.7474747474747	4484.70921644524\\
26.2424242424242	4773.96718254127\\
26.7373737373737	5497.84923626115\\
27.2323232323232	5507.20115259507\\
27.7272727272727	4645.66454316755\\
28.2222222222222	4652.18435759839\\
28.7171717171717	5864.77671493438\\
29.2121212121212	5217.22666484931\\
29.7070707070707	5225.2744382596\\
30.2020202020202	4525.19771426737\\
30.6969696969697	5598.97683636429\\
31.1919191919192	5607.85808936421\\
31.6868686868687	5616.85750947449\\
32.1818181818182	5625.96592088592\\
32.6767676767677	5635.17490172243\\
33.1717171717172	4970.93003020592\\
33.6666666666667	4978.50232131079\\
34.1616161616162	4986.18848227552\\
34.6565656565657	4993.98136466135\\
35.1515151515151	5001.87440134681\\
35.6464646464646	4824.1161645744\\
36.1414141414141	4831.59425151362\\
36.6363636363636	5588.28663588899\\
37.1313131313131	5597.40638437313\\
37.6262626262626	4903.37418583614\\
38.1212121212121	5413.39711888088\\
38.6161616161616	5492.79192697093\\
39.1111111111111	5500.3614812922\\
39.6060606060606	5307.27406694997\\
40.1010101010101	5313.91811116916\\
40.5959595959596	5320.58869167394\\
41.0909090909091	5327.28400143445\\
41.5858585858586	5334.00221693337\\
42.0808080808081	5394.10383458327\\
42.5757575757576	5400.65994218837\\
43.0707070707071	5060.45717704256\\
43.5656565656566	inf\\
44.0606060606061	inf\\
44.5555555555556	inf\\
45.0505050505051	inf\\
45.5454545454545	inf\\
46.040404040404	inf\\
46.5353535353535	inf\\
47.030303030303	inf\\
47.5252525252525	inf\\
48.020202020202	inf\\
48.5151515151515	inf\\
49.010101010101	inf\\
49.5050505050505	inf\\
50	inf\\
};
\addlegendentry{Tot}

%\addplot[only marks, mark=o, mark options={}, mark size=1.7678pt, draw=red] table[row sep=crcr]{%
%x	y\\
%17.8282828282828	4118.25041011359\\
%};
%\addlegendentry{optimal s}

\addplot [color=colorred, dashdotted, line width=1.5pt]
  table[row sep=crcr]{%
1	inf\\
1.4949494949495	inf\\
1.98989898989899	inf\\
2.48484848484848	inf\\
2.97979797979798	inf\\
3.47474747474747	inf\\
3.96969696969697	inf\\
4.46464646464646	1988.97954587334\\
4.95959595959596	1813.02134967086\\
5.45454545454545	1665.27777777778\\
5.94949494949495	1539.55799735386\\
6.44444444444444	1431.33174791914\\
6.93939393939394	1337.21706298507\\
7.43434343434343	1254.64282549622\\
7.92929292929293	1181.62197249381\\
8.42424242424242	1116.59593188758\\
8.91919191919192	1058.32581965373\\
9.41414141414141	1005.81494409549\\
9.90909090909091	958.252672333979\\
10.4040404040404	914.973136016589\\
10.8989898989899	875.42441814023\\
11.3939393939394	839.14525803531\\
11.8888888888889	805.747226832038\\
12.3838383838384	774.900936479592\\
12.8787878787879	746.325259515571\\
13.3737373737374	719.778821843539\\
13.8686868686869	695.05322977897\\
14.3636363636364	671.967633392085\\
14.8585858585859	650.364328987204\\
15.3535353535354	630.105176592798\\
15.8484848484848	611.068661823792\\
16.3434343434343	593.14747104958\\
16.8383838383838	576.246478358797\\
17.3333333333333	560.281065088758\\
17.8282828282828	545.17570961969\\
18.3232323232323	530.862798107091\\
18.8181818181818	517.281616840533\\
19.3131313131313	504.377494704224\\
19.8080808080808	492.101070310109\\
20.3030303030303	480.407663176654\\
20.7979797979798	469.256732133013\\
21.2929292929293	458.611407163098\\
21.7878787878788	448.438083337041\\
22.2828282828283	438.706067438265\\
22.7777777777778	429.387269482451\\
23.2727272727273	420.455932617188\\
23.7676767676768	411.888395947772\\
24.2626262626263	403.662885702381\\
24.7575757575758	395.759330865378\\
25.2525252525253	388.1592\\
25.7474747474747	380.845356474073\\
26.2424242424242	373.80192971321\\
26.7373737373737	367.01420044989\\
27.2323232323232	360.468498225748\\
27.7272727272727	354.152109647944\\
28.2222222222222	348.053196106392\\
28.7171717171717	342.160719833312\\
29.2121212121212	336.464377335101\\
29.7070707070707	330.954539353286\\
30.2020202020202	325.622196619725\\
30.6969696969697	320.458910764211\\
31.1919191919192	315.456769812612\\
31.6868686868687	310.608347782577\\
32.1818181818182	305.906667943439\\
32.6767676767677	301.345169358518\\
33.1717171717172	296.917676372802\\
33.6666666666667	292.618370747966\\
34.1616161616162	288.441766180646\\
34.6565656565657	284.382684969564\\
35.1515151515151	280.436236623068\\
35.6464646464646	276.597798221448\\
36.1414141414141	272.862996368404\\
36.6363636363636	269.22769058365\\
37.1313131313131	265.687958004218\\
37.6262626262626	262.240079275708\\
38.1212121212121	258.88052552693\\
38.6161616161616	255.605946332101\\
39.1111111111111	252.41315857438\\
39.6060606060606	249.299136132964\\
40.1010101010101	246.261000323586\\
40.5959595959596	243.296011028976\\
41.0909090909091	240.4015584619\\
41.5858585858586	237.575155508753\\
42.0808080808081	234.814430606545\\
42.5757575757576	232.117121110422\\
43.0707070707071	229.481067112771\\
43.5656565656566	inf\\
44.0606060606061	inf\\
44.5555555555556	inf\\
45.0505050505051	inf\\
45.5454545454545	inf\\
46.040404040404	inf\\
46.5353535353535	inf\\
47.030303030303	inf\\
47.5252525252525	inf\\
48.020202020202	inf\\
48.5151515151515	inf\\
49.010101010101	inf\\
49.5050505050505	inf\\
50	inf\\
};
\addlegendentry{Enc}

\addplot [color=colorgreen, dotted, line width=1.5pt]
  table[row sep=crcr]{%
1	inf\\
1.4949494949495	inf\\
1.98989898989899	inf\\
2.48484848484848	inf\\
2.97979797979798	inf\\
3.47474747474747	inf\\
3.96969696969697	inf\\
4.46464646464646	391.191477889997\\
4.95959595959596	468.499585496551\\
5.45454545454545	534.487573382061\\
5.94949494949495	530.451428660441\\
6.44444444444444	631.493843637034\\
6.93939393939394	598.325845564147\\
7.43434343434343	601.911402219715\\
7.92929292929293	566.881786857024\\
8.42424242424242	535.686990455248\\
8.91919191919192	507.732872567429\\
9.41414141414141	482.541415495779\\
9.90909090909091	517.122743516734\\
10.4040404040404	438.960688754674\\
10.8989898989899	419.987300516145\\
11.3939393939394	402.582395815456\\
11.8888888888889	386.559648696336\\
12.3838383838384	371.761054593655\\
12.8787878787879	358.051753902087\\
13.3737373737374	345.315958566265\\
13.8686868686869	333.453697494459\\
14.3636363636364	322.37821439671\\
14.8585858585859	291.00138466609\\
15.3535353535354	281.936514669484\\
15.8484848484848	273.418704358103\\
16.3434343434343	265.399938360332\\
16.8383838383838	257.837652839819\\
17.3333333333333	250.693986845744\\
17.8282828282828	204.752213955539\\
18.3232323232323	202.907730875145\\
18.8181818181818	191.088842994566\\
19.3131313131313	177.988081563433\\
19.8080808080808	173.655980128776\\
20.3030303030303	174.406763909722\\
20.7979797979798	171.18484802704\\
21.2929292929293	167.301613262599\\
21.7878787878788	135.650011530576\\
22.2828282828283	132.706057341246\\
22.7777777777778	130.401223247701\\
23.2727272727273	127.688862383553\\
23.7676767676768	125.086982185756\\
24.2626262626263	121.973628200312\\
24.7575757575758	121.531248995992\\
25.2525252525253	108.502261464907\\
25.7474747474747	104.826724045838\\
26.2424242424242	101.199929476803\\
26.7373737373737	100.586437569615\\
27.2323232323232	98.7925619173401\\
27.7272727272727	87.7701025254381\\
28.2222222222222	86.2586284943334\\
28.7171717171717	83.7748528138903\\
29.2121212121212	77.0025903168728\\
29.7070707070707	75.7417047794539\\
30.2020202020202	68.6805887361287\\
30.6969696969697	68.0368082853803\\
31.1919191919192	66.9748594075774\\
31.6868686868687	65.9455421432989\\
32.1818181818182	64.9473756450704\\
32.6767676767677	63.9789671510684\\
33.1717171717172	60.0405036214228\\
33.6666666666667	59.1712056781963\\
34.1616161616162	58.3267129098282\\
34.6565656565657	57.50597882644\\
35.1515151515151	56.7080150535686\\
35.6464646464646	53.971822536426\\
36.1414141414141	53.2431044718906\\
36.6363636363636	53.6840398980942\\
37.1313131313131	52.9783051897254\\
37.6262626262626	49.1328002532021\\
38.1212121212121	44.3328964000053\\
38.6161616161616	40.4678656481944\\
39.1111111111111	39.9625407079281\\
39.6060606060606	34.598922486549\\
40.1010101010101	34.1774271616666\\
40.5959595959596	33.7660737490104\\
41.0909090909091	33.3645006998612\\
41.5858585858586	32.972363247223\\
42.0808080808081	30.370519790629\\
42.5757575757576	30.0218596897616\\
43.0707070707071	27.1067790467181\\
43.5656565656566	inf\\
44.0606060606061	inf\\
44.5555555555556	inf\\
45.0505050505051	inf\\
45.5454545454545	inf\\
46.040404040404	inf\\
46.5353535353535	inf\\
47.030303030303	inf\\
47.5252525252525	inf\\
48.020202020202	inf\\
48.5151515151515	inf\\
49.010101010101	inf\\
49.5050505050505	inf\\
50	inf\\
};
\addlegendentry{Comm in}

\addplot [color=black, densely dotted, line width=1.5pt]
  table[row sep=crcr]{%
1	inf\\
1.4949494949495	inf\\
1.98989898989899	inf\\
2.48484848484848	inf\\
2.97979797979798	inf\\
3.47474747474747	inf\\
3.96969696969697	inf\\
4.46464646464646	4852.11077483074\\
4.95959595959596	3883.32511038549\\
5.45454545454545	4017.49194678511\\
5.94949494949495	2624.36967710965\\
6.44444444444444	3803.72293517579\\
6.93939393939394	3595.19740349931\\
7.43434343434343	3464.46796404399\\
7.92929292929293	3535.19710441207\\
8.42424242424242	3599.5036357695\\
8.91919191919192	3658.22431091203\\
9.41414141414141	3712.05667065297\\
9.90909090909091	3566.64180203784\\
10.4040404040404	3807.31038689318\\
10.8989898989899	3849.65019645198\\
11.3939393939394	3888.96843771528\\
11.8888888888889	3925.57735607846\\
12.3838383838384	3959.74774077537\\
12.8787878787879	3991.7153114173\\
13.3737373737374	4021.68650633125\\
13.8686868686869	4049.84260852202\\
14.3636363636364	4076.34372472631\\
14.8585858585859	4398.60097831599\\
15.3535353535354	4425.39632844433\\
15.8484848484848	4450.76075005258\\
16.3434343434343	4474.80595038301\\
16.8383838383838	4497.63217300305\\
17.3333333333333	4519.32988905314\\
17.8282828282828	2487.71950516511\\
18.3232323232323	2714.37944797504\\
18.8181818181818	2743.8589950905\\
19.3131313131313	3103.62555987061\\
19.8080808080808	3112.12174791524\\
20.3030303030303	3740.57488841474\\
20.7979797979798	4532.36476906126\\
21.2929292929293	4546.78533783975\\
21.7878787878788	3011.06208412498\\
22.2828282828283	3017.37950039797\\
22.7777777777778	3382.28592863834\\
23.2727272727273	3389.38002917543\\
23.7676767676768	3396.2002580554\\
24.2626262626263	3891.04313152785\\
24.7575757575758	4755.70324249893\\
25.2525252525253	3135.05821673363\\
25.7474747474747	3216.01571234441\\
26.2424242424242	3518.05871423818\\
26.7373737373737	4231.02261554955\\
27.2323232323232	4238.72360625448\\
27.7272727272727	3455.73288418408\\
28.2222222222222	3460.72081966881\\
28.7171717171717	4680.78940452756\\
29.2121212121212	4082.22320348659\\
29.7070707070707	4088.29012180509\\
30.2020202020202	3444.97849401896\\
30.6969696969697	4512.20816624464\\
31.1919191919192	4518.79117064075\\
31.6868686868687	4525.1835294564\\
32.1818181818182	4531.39342644253\\
32.6767676767677	4537.42852275727\\
33.1717171717172	3905.74222005061\\
33.6666666666667	3910.08958600969\\
34.1616161616162	3914.31837416369\\
34.6565656565657	3918.43338605531\\
35.1515151515151	3922.43919513034\\
35.6464646464646	3767.46080086146\\
36.1414141414141	3770.98785309926\\
36.6363636363636	4506.96069135771\\
37.1313131313131	4511.58192579513\\
37.6262626262626	3860.75597195841\\
38.1212121212121	4431.29469959154\\
38.6161616161616	4559.60296644168\\
39.1111111111111	4563.96421508933\\
39.6060606060606	4448.81105601195\\
40.1010101010101	4452.894536676\\
40.5959595959596	4456.88436713003\\
41.0909090909091	4460.78379648483\\
41.5858585858586	4464.59577890354\\
42.0808080808081	4563.19475247276\\
42.5757575757576	4567.07770419429\\
43.0707070707071	4274.43522776462\\
43.5656565656566	inf\\
44.0606060606061	inf\\
44.5555555555556	inf\\
45.0505050505051	inf\\
45.5454545454545	inf\\
46.040404040404	inf\\
46.5353535353535	inf\\
47.030303030303	inf\\
47.5252525252525	inf\\
48.020202020202	inf\\
48.5151515151515	inf\\
49.010101010101	inf\\
49.5050505050505	inf\\
50	inf\\
};
\addlegendentry{Comp}

\addplot [color=mycolor1, line width=1.5pt]
  table[row sep=crcr]{%
1	inf\\
1.4949494949495	inf\\
1.98989898989899	inf\\
2.48484848484848	inf\\
2.97979797979798	inf\\
3.47474747474747	inf\\
3.96969696969697	inf\\
4.46464646464646	77.9764635103574\\
4.95959595959596	115.239616948367\\
5.45454545454545	159.021096212844\\
5.94949494949495	187.761187717895\\
6.44444444444444	262.264850616664\\
6.93939393939394	288.124937256469\\
7.43434343434343	332.673195527814\\
7.92929292929293	356.419476692143\\
8.42424242424242	380.165595687267\\
8.91919191919192	403.911576042471\\
9.41414141414141	427.657437499852\\
9.90909090909091	507.763249233249\\
10.4040404040404	475.148856953202\\
10.8989898989899	498.894433976347\\
11.3939393939394	522.639933799871\\
11.8888888888889	546.385360237574\\
12.3838383838384	570.130720226928\\
12.8787878787879	593.876015138333\\
13.3737373737374	617.62125067203\\
13.8686868686869	641.36642720542\\
14.3636363636364	665.111549107394\\
14.8585858585859	642.465796548572\\
15.3535353535354	664.611900308515\\
15.8484848484848	686.757987000621\\
16.3434343434343	708.904059006259\\
16.8383838383838	731.050114541773\\
17.3333333333333	753.196156034325\\
17.8282828282828	650.800138475303\\
18.3232323232323	681.244125688021\\
18.8181818181818	676.691391196212\\
19.3131313131313	663.890277774747\\
19.8080808080808	681.356717029681\\
20.3030303030303	718.927422581031\\
20.7979797979798	740.470269052673\\
21.2929292929293	758.526840115226\\
21.7878787878788	643.946424342133\\
22.2828282828283	658.918003125705\\
22.7777777777778	676.556963825266\\
23.2727272727273	691.588205385826\\
23.7676767676768	706.619438172123\\
24.2626262626263	718.028296254496\\
24.7575757575758	744.91066906415\\
25.2525252525253	691.908105454206\\
25.7474747474747	694.93037612708\\
26.2424242424242	696.928322412356\\
26.7373737373737	719.079521225588\\
27.2323232323232	732.645074707669\\
27.7272727272727	674.777998961064\\
28.2222222222222	687.044651350668\\
28.7171717171717	690.871064504987\\
29.2121212121212	657.100084215873\\
29.7070707070707	668.428156716147\\
30.2020202020202	626.478248505116\\
30.6969696969697	641.113530958681\\
31.1919191919192	651.622399496633\\
31.6868686868687	662.13125886335\\
32.1818181818182	672.640109614682\\
32.6767676767677	683.148950620896\\
33.1717171717172	660.663378863176\\
33.6666666666667	670.672743470312\\
34.1616161616162	680.682082881782\\
34.6565656565657	690.691397628723\\
35.1515151515151	700.700689220219\\
35.6464646464646	685.803995866611\\
36.1414141414141	695.462397581797\\
36.6363636363636	720.559606265255\\
37.1313131313131	730.430126057989\\
37.6262626262626	695.590563731596\\
38.1212121212121	644.257537726153\\
38.6161616161616	603.460024871095\\
39.1111111111111	611.298598009275\\
39.6060606060606	542.732550401514\\
40.1010101010101	549.604134019296\\
40.5959595959596	556.475678753841\\
41.0909090909091	563.347186031772\\
41.5858585858586	570.218653847016\\
42.0808080808081	537.799466355851\\
42.5757575757576	544.204789477333\\
43.0707070707071	502.854010985205\\
43.5656565656566	inf\\
44.0606060606061	inf\\
44.5555555555556	inf\\
45.0505050505051	inf\\
45.5454545454545	inf\\
46.040404040404	inf\\
46.5353535353535	inf\\
47.030303030303	inf\\
47.5252525252525	inf\\
48.020202020202	inf\\
48.5151515151515	inf\\
49.010101010101	inf\\
49.5050505050505	inf\\
50	inf\\
};
\addlegendentry{Comm out}

\addplot [color=colorblue, loosely dashed, line width=1.5pt]
  table[row sep=crcr]{%
1	inf\\
1.4949494949495	inf\\
1.98989898989899	inf\\
2.48484848484848	inf\\
2.97979797979798	inf\\
3.47474747474747	inf\\
3.96969696969697	inf\\
4.46464646464646	9935.05137041418\\
4.95959595959596	7540.00772573099\\
5.45454545454545	5855.84538631152\\
5.94949494949495	4638.37807746695\\
6.44444444444444	3737.04719347322\\
6.93939393939394	3055.79083712842\\
7.43434343434343	2531.44361496317\\
7.92929292929293	2121.35187460184\\
8.42424242424242	1796.02567914181\\
8.91919191919192	1534.64326767537\\
9.41414141414141	1322.22383424854\\
9.90909090909091	1147.80485388865\\
10.4040404040404	1003.24041103263\\
10.8989898989899	882.392700076426\\
11.3939393939394	780.577916221673\\
11.8888888888889	694.180034880324\\
12.3838383838384	620.377418992442\\
12.8787878787879	556.946524443942\\
13.3737373737374	502.11910282589\\
13.8686868686869	454.477055051162\\
14.3636363636364	412.874133400999\\
14.8585858585859	376.377023791962\\
15.3535353535354	344.220577176648\\
15.8484848484848	315.773481021268\\
16.3434343434343	290.511710844635\\
16.8383838383838	267.997833674954\\
17.3333333333333	247.864751721498\\
17.8282828282828	229.802842897945\\
18.3232323232323	213.549720200157\\
18.8181818181818	198.882024942251\\
19.3131313131313	185.608810487165\\
19.8080808080808	173.566177927152\\
20.3030303030303	162.612903368826\\
20.7979797979798	152.626855262118\\
21.2929292929293	143.502044723099\\
21.7878787878788	135.146185735314\\
22.2828282828283	127.478668155549\\
22.7777777777778	120.428866558846\\
23.2727272727273	113.934723578018\\
23.7676767676768	107.941558595939\\
24.2626262626263	102.401062233994\\
24.7575757575758	97.2704446481619\\
25.2525252525253	92.5117116494948\\
25.7474747474747	88.0910474538429\\
26.2424242424242	83.9782867007192\\
26.7373737373737	80.1464614665046\\
27.2323232323232	76.5714114898329\\
27.7272727272727	73.2314478490286\\
28.2222222222222	70.1070619781873\\
28.7171717171717	67.1806732546273\\
29.2121212121212	64.4364094948698\\
29.7070707070707	61.8599156056217\\
30.2020202020202	59.4381863874401\\
30.6969696969697	57.1594201113829\\
31.1919191919192	55.0128900066363\\
31.6868686868687	52.988831228868\\
32.1818181818182	51.0783412402016\\
32.6767676767677	49.2732918346767\\
33.1717171717172	47.5662512979084\\
33.6666666666667	45.9504154046288\\
34.1616161616162	44.4195461395749\\
34.6565656565657	42.9679171813188\\
35.1515151515151	41.5902653196173\\
35.6464646464646	40.2817470884584\\
36.1414141414141	39.0378999922661\\
36.6363636363636	37.854607784279\\
37.1313131313131	36.7280693260621\\
37.6262626262626	35.6547706172316\\
38.1212121212121	34.6314596362526\\
38.6161616161616	33.6551236778587\\
39.1111111111111	32.7229689112868\\
39.6060606060606	31.8324019169975\\
40.1010101010101	30.9810129886151\\
40.5959595959596	30.1665610120876\\
41.0909090909091	29.3869597560826\\
41.5858585858586	28.6402654268389\\
42.0808080808081	27.9246653574883\\
42.5757575757576	27.2384677165552\\
43.0707070707071	26.5800921332363\\
43.5656565656566	inf\\
44.0606060606061	inf\\
44.5555555555556	inf\\
45.0505050505051	inf\\
45.5454545454545	inf\\
46.040404040404	inf\\
46.5353535353535	inf\\
47.030303030303	inf\\
47.5252525252525	inf\\
48.020202020202	inf\\
48.5151515151515	inf\\
49.010101010101	inf\\
49.5050505050505	inf\\
50	inf\\
};
\addlegendentry{Dec}

\addplot[only marks, mark=*, mark options={}, mark size=3pt, draw=red, fill=red] table[row sep=crcr]{%
x	y\\
17.8282828282828	4118.25\\
};

\addplot[only marks, mark=triangle*, mark options={}, mark size=4pt,  draw=myblue, fill=myblue] table[row sep=crcr]{%
x	y\\
43	5060.45\\
};

\end{axis}

\begin{axis}[%
width=5cm,
height=8cm,
at={(0cm,0.7cm)},
scale only axis,
xmin=1,
xmax=50,
xlabel style={font=\color{white!15!black}},
xlabel={s},
xtick={1,5,15,43,50},
xticklabels={{1},{5},{15},{43},{50}},
ymin=0,
ymax=151,
ytick={0,50,100,150},
yticklabels={{0},{50},{100},{150}},
axis background/.style={fill=white},
title style={font=\bfseries},
title={Number of valid workers},
xmajorgrids,
ymajorgrids,
legend style={legend cell align=left, align=left, draw=white!15!black}
]
\addplot [color=colorred, dashdotted,line width=1.5pt]
  table[row sep=crcr]{%
1	59\\
1.4949494949495	88\\
1.98989898989899	119\\
2.48484848484848	149\\
2.97979797979798	150\\
3.47474747474747	150\\
3.96969696969697	150\\
4.46464646464646	150\\
4.95959595959596	150\\
5.45454545454545	150\\
5.94949494949495	150\\
6.44444444444444	150\\
6.93939393939394	150\\
7.43434343434343	150\\
7.92929292929293	150\\
8.42424242424242	150\\
8.91919191919192	150\\
9.41414141414141	150\\
9.90909090909091	150\\
10.4040404040404	150\\
10.8989898989899	150\\
11.3939393939394	150\\
11.8888888888889	150\\
12.3838383838384	150\\
12.8787878787879	150\\
13.3737373737374	150\\
13.8686868686869	150\\
14.3636363636364	150\\
14.8585858585859	150\\
15.3535353535354	150\\
15.8484848484848	150\\
16.3434343434343	150\\
16.8383838383838	150\\
17.3333333333333	150\\
17.8282828282828	150\\
18.3232323232323	150\\
18.8181818181818	150\\
19.3131313131313	150\\
19.8080808080808	150\\
20.3030303030303	150\\
20.7979797979798	150\\
21.2929292929293	150\\
21.7878787878788	150\\
22.2828282828283	150\\
22.7777777777778	150\\
23.2727272727273	150\\
23.7676767676768	150\\
24.2626262626263	150\\
24.7575757575758	150\\
25.2525252525253	150\\
25.7474747474747	150\\
26.2424242424242	150\\
26.7373737373737	150\\
27.2323232323232	150\\
27.7272727272727	150\\
28.2222222222222	150\\
28.7171717171717	150\\
29.2121212121212	150\\
29.7070707070707	150\\
30.2020202020202	150\\
30.6969696969697	150\\
31.1919191919192	150\\
31.6868686868687	150\\
32.1818181818182	150\\
32.6767676767677	150\\
33.1717171717172	150\\
33.6666666666667	150\\
34.1616161616162	150\\
34.6565656565657	150\\
35.1515151515151	150\\
35.6464646464646	150\\
36.1414141414141	150\\
36.6363636363636	150\\
37.1313131313131	150\\
37.6262626262626	150\\
38.1212121212121	150\\
38.6161616161616	150\\
39.1111111111111	150\\
39.6060606060606	150\\
40.1010101010101	150\\
40.5959595959596	150\\
41.0909090909091	150\\
41.5858585858586	150\\
42.0808080808081	150\\
42.5757575757576	150\\
43.0707070707071	150\\
43.5656565656566	150\\
44.0606060606061	150\\
44.5555555555556	150\\
45.0505050505051	150\\
45.5454545454545	150\\
46.040404040404	150\\
46.5353535353535	150\\
47.030303030303	150\\
47.5252525252525	150\\
48.020202020202	150\\
48.5151515151515	150\\
49.010101010101	150\\
49.5050505050505	150\\
50	150\\
};
%\addlegendentry{Enc}

\addplot [color=colorgreen, dotted, line width=1.5pt]
  table[row sep=crcr]{%
1	12\\
1.4949494949495	17\\
1.98989898989899	28\\
2.48484848484848	36\\
2.97979797979798	45\\
3.47474747474747	50\\
3.96969696969697	60\\
4.46464646464646	65\\
4.95959595959596	76\\
5.45454545454545	86\\
5.94949494949495	92\\
6.44444444444444	93\\
6.93939393939394	94\\
7.43434343434343	99\\
7.92929292929293	102\\
8.42424242424242	104\\
8.91919191919192	107\\
9.41414141414141	109\\
9.90909090909091	111\\
10.4040404040404	113\\
10.8989898989899	116\\
11.3939393939394	117\\
11.8888888888889	118\\
12.3838383838384	121\\
12.8787878787879	123\\
13.3737373737374	124\\
13.8686868686869	124\\
14.3636363636364	125\\
14.8585858585859	125\\
15.3535353535354	125\\
15.8484848484848	126\\
16.3434343434343	126\\
16.8383838383838	127\\
17.3333333333333	127\\
17.8282828282828	127\\
18.3232323232323	130\\
18.8181818181818	130\\
19.3131313131313	130\\
19.8080808080808	130\\
20.3030303030303	130\\
20.7979797979798	130\\
21.2929292929293	130\\
21.7878787878788	132\\
22.2828282828283	132\\
22.7777777777778	133\\
23.2727272727273	134\\
23.7676767676768	134\\
24.2626262626263	135\\
24.7575757575758	135\\
25.2525252525253	135\\
25.7474747474747	135\\
26.2424242424242	137\\
26.7373737373737	138\\
27.2323232323232	138\\
27.7272727272727	138\\
28.2222222222222	138\\
28.7171717171717	138\\
29.2121212121212	138\\
29.7070707070707	139\\
30.2020202020202	139\\
30.6969696969697	139\\
31.1919191919192	140\\
31.6868686868687	140\\
32.1818181818182	140\\
32.6767676767677	140\\
33.1717171717172	140\\
33.6666666666667	141\\
34.1616161616162	141\\
34.6565656565657	141\\
35.1515151515151	141\\
35.6464646464646	141\\
36.1414141414141	141\\
36.6363636363636	141\\
37.1313131313131	141\\
37.6262626262626	141\\
38.1212121212121	141\\
38.6161616161616	141\\
39.1111111111111	141\\
39.6060606060606	141\\
40.1010101010101	141\\
40.5959595959596	141\\
41.0909090909091	141\\
41.5858585858586	141\\
42.0808080808081	141\\
42.5757575757576	141\\
43.0707070707071	141\\
43.5656565656566	141\\
44.0606060606061	141\\
44.5555555555556	141\\
45.0505050505051	141\\
45.5454545454545	141\\
46.040404040404	141\\
46.5353535353535	141\\
47.030303030303	141\\
47.5252525252525	141\\
48.020202020202	141\\
48.5151515151515	141\\
49.010101010101	141\\
49.5050505050505	141\\
50	141\\
};
%\addlegendentry{Comm in}

\addplot [color=mycolor1, line width=1.5pt]
  table[row sep=crcr]{%
1	146\\
1.4949494949495	142\\
1.98989898989899	141\\
2.48484848484848	141\\
2.97979797979798	140\\
3.47474747474747	138\\
3.96969696969697	135\\
4.46464646464646	132\\
4.95959595959596	130\\
5.45454545454545	130\\
5.94949494949495	127\\
6.44444444444444	126\\
6.93939393939394	125\\
7.43434343434343	124\\
7.92929292929293	121\\
8.42424242424242	118\\
8.91919191919192	117\\
9.41414141414141	116\\
9.90909090909091	112\\
10.4040404040404	109\\
10.8989898989899	108\\
11.3939393939394	106\\
11.8888888888889	104\\
12.3838383838384	102\\
12.8787878787879	101\\
13.3737373737374	100\\
13.8686868686869	98\\
14.3636363636364	95\\
14.8585858585859	94\\
15.3535353535354	93\\
15.8484848484848	93\\
16.3434343434343	92\\
16.8383838383838	92\\
17.3333333333333	92\\
17.8282828282828	88\\
18.3232323232323	86\\
18.8181818181818	82\\
19.3131313131313	80\\
19.8080808080808	77\\
20.3030303030303	74\\
20.7979797979798	71\\
21.2929292929293	71\\
21.7878787878788	66\\
22.2828282828283	66\\
22.7777777777778	65\\
23.2727272727273	63\\
23.7676767676768	63\\
24.2626262626263	62\\
24.7575757575758	61\\
25.2525252525253	60\\
25.7474747474747	58\\
26.2424242424242	57\\
26.7373737373737	56\\
27.2323232323232	56\\
27.7272727272727	54\\
28.2222222222222	53\\
28.7171717171717	51\\
29.2121212121212	49\\
29.7070707070707	49\\
30.2020202020202	48\\
30.6969696969697	47\\
31.1919191919192	46\\
31.6868686868687	46\\
32.1818181818182	46\\
32.6767676767677	46\\
33.1717171717172	45\\
33.6666666666667	45\\
34.1616161616162	44\\
34.6565656565657	44\\
35.1515151515151	44\\
35.6464646464646	43\\
36.1414141414141	43\\
36.6363636363636	42\\
37.1313131313131	42\\
37.6262626262626	41\\
38.1212121212121	38\\
38.6161616161616	37\\
39.1111111111111	37\\
39.6060606060606	36\\
40.1010101010101	36\\
40.5959595959596	36\\
41.0909090909091	36\\
41.5858585858586	36\\
42.0808080808081	35\\
42.5757575757576	35\\
43.0707070707071	34\\
43.5656565656566	32\\
44.0606060606061	32\\
44.5555555555556	32\\
45.0505050505051	32\\
45.5454545454545	32\\
46.040404040404	31\\
46.5353535353535	30\\
47.030303030303	30\\
47.5252525252525	30\\
48.020202020202	29\\
48.5151515151515	28\\
49.010101010101	28\\
49.5050505050505	28\\
50	28\\
};
%\addlegendentry{Comm out}

\addplot [color=colorblue,loosely dashed, line width=1.5pt]
  table[row sep=crcr]{%
1	0\\
1.4949494949495	1\\
1.98989898989899	3\\
2.48484848484848	12\\
2.97979797979798	22\\
3.47474747474747	28\\
3.96969696969697	38\\
4.46464646464646	49\\
4.95959595959596	69\\
5.45454545454545	93\\
5.94949494949495	117\\
6.44444444444444	148\\
6.93939393939394	150\\
7.43434343434343	150\\
7.92929292929293	150\\
8.42424242424242	150\\
8.91919191919192	150\\
9.41414141414141	150\\
9.90909090909091	150\\
10.4040404040404	150\\
10.8989898989899	150\\
11.3939393939394	150\\
11.8888888888889	150\\
12.3838383838384	150\\
12.8787878787879	150\\
13.3737373737374	150\\
13.8686868686869	150\\
14.3636363636364	150\\
14.8585858585859	150\\
15.3535353535354	150\\
15.8484848484848	150\\
16.3434343434343	150\\
16.8383838383838	150\\
17.3333333333333	150\\
17.8282828282828	150\\
18.3232323232323	150\\
18.8181818181818	150\\
19.3131313131313	150\\
19.8080808080808	150\\
20.3030303030303	150\\
20.7979797979798	150\\
21.2929292929293	150\\
21.7878787878788	150\\
22.2828282828283	150\\
22.7777777777778	150\\
23.2727272727273	150\\
23.7676767676768	150\\
24.2626262626263	150\\
24.7575757575758	150\\
25.2525252525253	150\\
25.7474747474747	150\\
26.2424242424242	150\\
26.7373737373737	150\\
27.2323232323232	150\\
27.7272727272727	150\\
28.2222222222222	150\\
28.7171717171717	150\\
29.2121212121212	150\\
29.7070707070707	150\\
30.2020202020202	150\\
30.6969696969697	150\\
31.1919191919192	150\\
31.6868686868687	150\\
32.1818181818182	150\\
32.6767676767677	150\\
33.1717171717172	150\\
33.6666666666667	150\\
34.1616161616162	150\\
34.6565656565657	150\\
35.1515151515151	150\\
35.6464646464646	150\\
36.1414141414141	150\\
36.6363636363636	150\\
37.1313131313131	150\\
37.6262626262626	150\\
38.1212121212121	150\\
38.6161616161616	150\\
39.1111111111111	150\\
39.6060606060606	150\\
40.1010101010101	150\\
40.5959595959596	150\\
41.0909090909091	150\\
41.5858585858586	150\\
42.0808080808081	150\\
42.5757575757576	150\\
43.0707070707071	150\\
43.5656565656566	150\\
44.0606060606061	150\\
44.5555555555556	150\\
45.0505050505051	150\\
45.5454545454545	150\\
46.040404040404	150\\
46.5353535353535	150\\
47.030303030303	150\\
47.5252525252525	150\\
48.020202020202	150\\
48.5151515151515	150\\
49.010101010101	150\\
49.5050505050505	150\\
50	150\\
};
%\addlegendentry{dec}

\addplot [color=colorbluedark, dashed, line width=1.5pt]
  table[row sep=crcr]{%
1	0\\
1.4949494949495	0\\
1.98989898989899	2\\
2.48484848484848	12\\
2.97979797979798	19\\
3.47474747474747	24\\
3.96969696969697	30\\
4.46464646464646	38\\
4.95959595959596	53\\
5.45454545454545	67\\
5.94949494949495	79\\
6.44444444444444	92\\
6.93939393939394	94\\
7.43434343434343	99\\
7.92929292929293	102\\
8.42424242424242	104\\
8.91919191919192	107\\
9.41414141414141	109\\
9.90909090909091	111\\
10.4040404040404	109\\
10.8989898989899	108\\
11.3939393939394	106\\
11.8888888888889	104\\
12.3838383838384	102\\
12.8787878787879	101\\
13.3737373737374	100\\
13.8686868686869	98\\
14.3636363636364	95\\
14.8585858585859	94\\
15.3535353535354	93\\
15.8484848484848	93\\
16.3434343434343	92\\
16.8383838383838	92\\
17.3333333333333	92\\
17.8282828282828	88\\
18.3232323232323	86\\
18.8181818181818	82\\
19.3131313131313	80\\
19.8080808080808	77\\
20.3030303030303	74\\
20.7979797979798	71\\
21.2929292929293	71\\
21.7878787878788	66\\
22.2828282828283	66\\
22.7777777777778	65\\
23.2727272727273	63\\
23.7676767676768	63\\
24.2626262626263	62\\
24.7575757575758	61\\
25.2525252525253	60\\
25.7474747474747	58\\
26.2424242424242	57\\
26.7373737373737	56\\
27.2323232323232	56\\
27.7272727272727	54\\
28.2222222222222	53\\
28.7171717171717	51\\
29.2121212121212	49\\
29.7070707070707	49\\
30.2020202020202	48\\
30.6969696969697	47\\
31.1919191919192	46\\
31.6868686868687	46\\
32.1818181818182	46\\
32.6767676767677	46\\
33.1717171717172	45\\
33.6666666666667	45\\
34.1616161616162	44\\
34.6565656565657	44\\
35.1515151515151	44\\
35.6464646464646	43\\
36.1414141414141	43\\
36.6363636363636	42\\
37.1313131313131	42\\
37.6262626262626	41\\
38.1212121212121	38\\
38.6161616161616	37\\
39.1111111111111	37\\
39.6060606060606	36\\
40.1010101010101	36\\
40.5959595959596	36\\
41.0909090909091	36\\
41.5858585858586	36\\
42.0808080808081	35\\
42.5757575757576	35\\
43.0707070707071	34\\
43.5656565656566	32\\
44.0606060606061	32\\
44.5555555555556	32\\
45.0505050505051	32\\
45.5454545454545	32\\
46.040404040404	31\\
46.5353535353535	30\\
47.030303030303	30\\
47.5252525252525	30\\
48.020202020202	29\\
48.5151515151515	28\\
49.010101010101	28\\
49.5050505050505	28\\
50	28\\
};
%\addlegendentry{tot}

\end{axis}

\begin{axis}[%
width=5cm,
height=8cm,
at={(6.5cm,0.7cm)},
scale only axis,
unbounded coords=jump,
xmin=1,
xmax=50,
xlabel style={font=\color{white!15!black}},
xlabel={s},
xtick={1,5,15,43,50},
xticklabels={{1},{5},{15},{43},{50}},
ymin=1,
ymax=1.5,
ytick={1,1.2,1.5},
yticklabels={{1},{$1+\Theta$},{1.5}},
axis background/.style={fill=white},
title style={font=\bfseries},
title={$\sum_{p=1}^Pr_p^{\text{comp}}$},
xmajorgrids,
ymajorgrids
]
\addplot [color=colorbluedark, dashed, line width=1.5pt, forget plot]
  table[row sep=crcr]{%
1	inf\\
1.4949494949495	inf\\
1.98989898989899	inf\\
2.48484848484848	inf\\
2.97979797979798	inf\\
3.47474747474747	inf\\
3.96969696969697	inf\\
4.46464646464646	1.20086298389942\\
4.95959595959596	1.25695273635441\\
5.45454545454545	1.24877546364646\\
5.94949494949495	1.38087975518324\\
6.44444444444444	1.26288257205442\\
6.93939393939394	1.27808153703725\\
7.43434343434343	1.28857996486379\\
7.92929292929293	1.28280629214636\\
8.42424242424242	1.27775386853832\\
8.91919191919192	1.27329549113734\\
9.41414141414141	1.26933220833726\\
9.90909090909091	1.28034835060846\\
10.4040404040404	1.26259402126156\\
10.8989898989899	1.25970599304674\\
11.3939393939394	1.25708039113239\\
11.8888888888889	1.25468300118043\\
12.3838383838384	1.25248530881904\\
12.8787878787879	1.25046336011445\\
13.3737373737374	1.24859688490775\\
13.8686868686869	1.24686861489202\\
14.3636363636364	1.24526374776265\\
14.8585858585859	1.22729176989874\\
15.3535353535354	1.22591560207526\\
15.8484848484848	1.22462819005744\\
16.3434343434343	1.22342121564099\\
16.8383838383838	1.22228736857713\\
17.3333333333333	1.22122019838831\\
17.8282828282828	1.4016229679423\\
18.3232323232323	1.36792824782222\\
18.8181818181818	1.36391756163464\\
19.3131313131313	1.32142183167828\\
19.8080808080808	1.3205448701\\
20.3030303030303	1.26676207551443\\
20.7979797979798	1.22024556620008\\
21.2929292929293	1.2195472673943\\
21.7878787878788	1.33177291123058\\
22.2828282828283	1.33107847583587\\
22.7777777777778	1.29534402161933\\
23.2727272727273	1.29472601999273\\
23.7676767676768	1.29413431059944\\
24.2626262626263	1.25659223545584\\
24.7575757575758	1.20999530711353\\
25.2525252525253	1.31828443980751\\
25.7474747474747	1.31014323998709\\
26.2424242424242	1.28337549591713\\
26.7373737373737	1.23562442106927\\
27.2323232323232	1.23519660024265\\
27.7272727272727	1.2880344464756\\
28.2222222222222	1.28761975684797\\
28.7171717171717	1.21219221658147\\
29.2121212121212	1.24334761616676\\
29.7070707070707	1.24298697744112\\
30.2020202020202	1.28778464949149\\
30.6969696969697	1.22027455250339\\
31.1919191919192	1.21995402234893\\
31.6868686868687	1.21964366610706\\
32.1818181818182	1.21934300695166\\
32.6767676767677	1.21905159739514\\
33.1717171717172	1.25415910998029\\
33.6666666666667	1.25387700305534\\
34.1616161616162	1.25360319214799\\
34.6565656565657	1.25333731661515\\
35.1515151515151	1.25307903641891\\
35.6464646464646	1.26360217455058\\
36.1414141414141	1.2633560415402\\
36.6363636363636	1.22000648206654\\
37.1313131313131	1.2197815102165\\
37.6262626262626	1.25571445692877\\
38.1212121212121	1.22033379814713\\
38.6161616161616	1.21182380855114\\
39.1111111111111	1.21162267294083\\
39.6060606060606	1.21081211210208\\
40.1010101010101	1.21062106531641\\
40.5959595959596	1.21043473514115\\
41.0909090909091	1.21025294903884\\
41.5858585858586	1.21007554278608\\
42.0808080808081	1.20058262146226\\
42.5757575757576	1.20041481567789\\
43.0707070707071	1.2064815167541\\
43.5656565656566	inf\\
44.0606060606061	inf\\
44.5555555555556	inf\\
45.0505050505051	inf\\
45.5454545454545	inf\\
46.040404040404	inf\\
46.5353535353535	inf\\
47.030303030303	inf\\
47.5252525252525	inf\\
48.020202020202	inf\\
48.5151515151515	inf\\
49.010101010101	inf\\
49.5050505050505	inf\\
50	inf\\
};
\addplot[only marks, mark=*, mark options={}, mark size=3pt, draw=red, fill=red] table[row sep=crcr]{%
x	y\\
17.8282828282828	1.4016229679423\\
};
 \node at (17.82,1.45) {\textcolor{red}{Delay choice}};

\addplot[only marks, mark=triangle*, mark options={}, mark size=4pt,  draw=myblue, fill=myblue] table[row sep=crcr]{%
x	y\\
43	1.2\\
};
 \node at (35,1.17) {\textcolor{myblue}{Quantization}};
 \node at (35,1.14) {\textcolor{myblue}{choice}};
\end{axis}
\end{tikzpicture}%

%% file: Comparisonheterogeneousworkers_withgini.tikz
\definecolor{ginipurg}{rgb}{0.55000,0.710000,0.00000}%
\definecolor{ginino}{rgb}{0.00000,0.26000,0.15000}%
\definecolor{unipurg}{rgb}{0.00000,0.00000,1.00000}%
\definecolor{unino}{rgb}{0.00000,0.50000,0.50000}
\definecolor{purg}{rgb}{0.640000,0.00000,0.00000}%
\definecolor{no}{rgb}{1.00000,0.20000,0.00000}%
\begin{tikzpicture}

\begin{axis}[%
width=9cm,
height=7cm,
at={(0.5cm,0.5cm)},
scale only axis,
xmin=0,
xmax=60000,
xlabel style={font=\color{white!15!black}},
xlabel={Average in-order execution time},
xtick={0,10000,30000,50000,60000},
xticklabels={{0},{$\frac{10}{\lambda}$},{$\frac{30}{\lambda}$},{$\frac{50}{\lambda}$},{$\frac{60}{\lambda}$}},
scaled x ticks = false,
ymin=1,
ymax=1.7,
ytick={1,1.1,1.4,1.7},
yticklabels={{1},{$1.1$},{$1.4$},{$1.7$}},
ylabel style={font=\color{white!15!black}},
ylabel={Computational load},
axis background/.style={fill=white},
title style={font=\bfseries},
title={Average},
xmajorgrids,
ymajorgrids,
legend style={legend cell align=left, align=left, draw=white!15!black}
]

\addplot [color=unino,line width=1.5pt, dashed, mark=triangle*, mark options={solid, fill=unino, unino}]
  table[row sep=crcr]{%
57488.0085195251	1.01419763037975\\
56070.6025989696	1.09049681012657\\
25279.021988414	1.14301974683545\\
11132.7252678585	1.20886618734177\\
5221.64555000833	1.27180459746835\\
2299.31971764917	1.34642292658228\\
1161.76333820794	1.42195823797468\\
907.961096357907	1.4991571443038\\
600.29809096053	1.5769993721519\\
515.613237185697	1.64535291139241\\
};
%\addlegendentry{uni without purging}

\addplot [color=unipurg,line width=1.5pt, mark=triangle*, mark options={solid, fill=unipurg, unipurg}]
  table[row sep=crcr]{%
47175.2392495251	1.00693738758611\\
13129.6691695251	1.01190935878813\\
3761.78254952515	1.01679311392405\\
1705.25072952515	1.02186458734177\\
988.750659525146	1.02615882531646\\
734.183859525146	1.0315613164557\\
617.356199525146	1.03669670886076\\
559.654310306709	1.04130681518987\\
527.325319525147	1.04596042531646\\
509.394639525146	1.04991594936709\\
};
%\addlegendentry{uni with purging}

\addplot [color=no,line width=1.5pt, mark=+, mark size = 3pt, dashed, mark options={solid, fill=no, no}]
  table[row sep=crcr]{%
813.721489525147	1.00520644050633\\
743.116628969592	1.08539511898734\\
688.758228414036	1.15899511898734\\
649.422867858479	1.23892731139241\\
613.852937302924	1.31306920506329\\
597.393465144163	1.39244208607595\\
564.449226191814	1.47259163544304\\
538.19842367233	1.54640032405063\\
512.706389540986	1.62462225822785\\
495.335199916267	1.6950950886076\\
};
%\addlegendentry{our without purging}

\addplot [color=purg,line width=1.5pt, mark=+,mark size = 3pt, mark options={solid, fill=purg, purg}]
  table[row sep=crcr]{%
816.259729525146	1.00505203037975\\
763.664609525146	1.03276565063291\\
705.709999525146	1.04198837468355\\
675.290949525146	1.04867264810127\\
645.782859525146	1.05324962025316\\
629.153299525147	1.05641940253165\\
615.261219525147	1.05955532151899\\
602.319679525147	1.06130738227848\\
603.665309525147	1.06311185822785\\
583.453759525147	1.06338478987342\\
};
%\addlegendentry{our with purging}
\addplot [color=ginino,line width=1.5pt, dashed, mark=*, mark options={solid, fill=ginino, ginino}]
  table[row sep=crcr]{%
592.963329525146	1.00455696202532\\
465.767028969592	1.08516998481012\\
441.432284866942	1.15771633417723\\
430.711314093298	1.23832668354431\\
417.046952404965	1.31070371645569\\
410.07402515073	1.39147997974683\\
401.080289106611	1.47196070886077\\
393.321006564332	1.54346693670887\\
389.416475080701	1.62414655189872\\
382.643290837831	1.69257778227847\\
};
%\addlegendentry{gini without purging}

\addplot [color=ginipurg,line width=1.5pt, mark=*, mark options={solid, fill=ginipurg, ginipurg}]
  table[row sep=crcr]{%
593.918109525146	1.00455696202532\\
500.292509525147	1.04517257721519\\
504.054509525146	1.05321073417722\\
505.303349525147	1.05338677468354\\
504.607899525146	1.05344437468354\\
502.911529525146	1.05346730126582\\
504.973389525146	1.05340864810127\\
502.308309525147	1.05354653164557\\
503.833329525147	1.05359278987342\\
502.321139525147	1.05340962025316\\
};
%\addlegendentry{gini with purging}
 \node at (12000,1.55) {\textcolor{ginino}{Ideal} and \textcolor{no}{Optimal} };
%\node at (800,1.3) {};
\node at (50000,1.15) {\textcolor{unino}{Uniform}};

\node at (12000,1.1) {\textcolor{ginipurg}{Ideal} and \textcolor{purg}{Optimal}};
\node at (12000,1.05) {(purging)};
\node at (40000,1.05) {\textcolor{unipurg}{Uniform (purging)}};

\end{axis}
\end{tikzpicture}%

%% file: Comparisonheterogeneousworkers_with_gini_zoom.tikz
\definecolor{ginipurg}{rgb}{0.55000,0.710000,0.00000}%
\definecolor{ginino}{rgb}{0.00000,0.26000,0.15000}%
\definecolor{unipurg}{rgb}{0.00000,0.00000,1.00000}%
\definecolor{unino}{rgb}{0.00000,0.50000,0.50000}
\definecolor{purg}{rgb}{0.640000,0.00000,0.00000}%
\definecolor{no}{rgb}{1.00000,0.20000,0.00000}%
\begin{tikzpicture}

\begin{axis}[%
width=9cm,
height=7cm,
at={(0.5cm,0.5cm)},
scale only axis,
xmin=0,
xmax=2000,
xlabel style={font=\color{white!15!black}},
xlabel={Average in-order execution time},
xtick={0,333,500,1000,2000},
xticklabels={{0},{$\frac{1}{3\lambda}$},{$\frac{1}{2\lambda}$},{$\frac{1}{\lambda}$},{$\frac{2}{\lambda}$}},
scaled x ticks = false,
ymin=1,
ymax=1.7,
ytick={1,1.1,1.4,1.7},
yticklabels={{1},{$1.1$},{$1.4$},{$1.7$}},
ylabel style={font=\color{white!15!black}},
ylabel={Computational load},
axis background/.style={fill=white},
title style={font=\bfseries},
title={Average (zoomed-in)},
xmajorgrids,
ymajorgrids,
legend style={legend cell align=left, align=left, draw=white!15!black}
]
\addplot [color=ginino,line width=1.5pt, dashed, mark=*, mark options={solid, fill=ginino, ginino}]
  table[row sep=crcr]{%
592.963329525146	1.00455696202532\\
465.767028969592	1.08516998481012\\
441.432284866942	1.15771633417723\\
430.711314093298	1.23832668354431\\
417.046952404965	1.31070371645569\\
410.07402515073	1.39147997974683\\
401.080289106611	1.47196070886077\\
393.321006564332	1.54346693670887\\
389.416475080701	1.62414655189872\\
382.643290837831	1.69257778227847\\
};
%\addlegendentry{gini without purging}

\addplot [color=unino,line width=1.5pt, dashed, mark=triangle*, mark options={solid, fill=unino, unino}]
  table[row sep=crcr]{%
57488.0085195251	1.01419763037975\\
56070.6025989696	1.09049681012657\\
25279.021988414	1.14301974683545\\
11132.7252678585	1.20886618734177\\
5221.64555000833	1.27180459746835\\
2299.31971764917	1.34642292658228\\
1161.76333820794	1.42195823797468\\
907.961096357907	1.4991571443038\\
600.29809096053	1.5769993721519\\
515.613237185697	1.64535291139241\\
};
%\addlegendentry{uni without purging}

\addplot [color=unipurg,line width=1.5pt, mark=triangle*, mark options={solid, fill=unipurg, unipurg}]
  table[row sep=crcr]{%
47175.2392495251	1.00693738758611\\
13129.6691695251	1.01190935878813\\
3761.78254952515	1.01679311392405\\
1705.25072952515	1.02186458734177\\
988.750659525146	1.02615882531646\\
734.183859525146	1.0315613164557\\
617.356199525146	1.03669670886076\\
559.654310306709	1.04130681518987\\
527.325319525147	1.04596042531646\\
509.394639525146	1.04991594936709\\
};
%\addlegendentry{uni with purging}

\addplot [color=no,line width=1.5pt, mark=+, mark size = 3pt, dashed, mark options={solid, fill=no, no}]
  table[row sep=crcr]{%
813.721489525147	1.00520644050633\\
743.116628969592	1.08539511898734\\
688.758228414036	1.15899511898734\\
649.422867858479	1.23892731139241\\
613.852937302924	1.31306920506329\\
597.393465144163	1.39244208607595\\
564.449226191814	1.47259163544304\\
538.19842367233	1.54640032405063\\
512.706389540986	1.62462225822785\\
495.335199916267	1.6950950886076\\
};
%\addlegendentry{our without purging}

\addplot [color=purg,line width=1.5pt, mark=+,mark size = 3pt, mark options={solid, fill=purg, purg}]
  table[row sep=crcr]{%
816.259729525146	1.00505203037975\\
763.664609525146	1.03276565063291\\
705.709999525146	1.04198837468355\\
675.290949525146	1.04867264810127\\
645.782859525146	1.05324962025316\\
629.153299525147	1.05641940253165\\
615.261219525147	1.05955532151899\\
602.319679525147	1.06130738227848\\
603.665309525147	1.06311185822785\\
583.453759525147	1.06338478987342\\
};

\addplot [color=ginipurg,line width=1.5pt, mark=*, mark options={solid, fill=ginipurg, ginipurg}]
  table[row sep=crcr]{%
593.918109525146	1.00455696202532\\
500.292509525147	1.04517257721519\\
504.054509525146	1.05321073417722\\
505.303349525147	1.05338677468354\\
504.607899525146	1.05344437468354\\
502.911529525146	1.05346730126582\\
504.973389525146	1.05340864810127\\
502.308309525147	1.05354653164557\\
503.833329525147	1.05359278987342\\
502.321139525147	1.05340962025316\\
};
%\addlegendentry{gini with purging}
%\addlegendentry{our with purging}

 \node at (200,1.55) {\textcolor{ginino}{Ideal}};
\node at (800,1.3) {\textcolor{no}{Optimal}};
\node at (1600,1.45) {\textcolor{unino}{Uniform}};

\node at (200,1.1) {\textcolor{ginipurg}{Ideal}};
\node at (200,1.05) {\textcolor{ginipurg}{(purging)}};
\node at (900,1.12) {\textcolor{purg}{Optimal (purging)}};
\node at (1600,1.07) {\textcolor{unipurg}{Uniform (purging)}};

\end{axis}
\end{tikzpicture}%

%% file: Comparisonheterogeneousworkers_delayOmega.tikz
\definecolor{mycolor1}{rgb}{1.00000,0.00000,1.00000}%
\definecolor{ginipurg}{rgb}{0.00000,0.26000,0.15000}%
\definecolor{ginino}{rgb}{0.55000,0.710000,0.00000}%
\definecolor{unipurg}{rgb}{0.00000,0.00000,1.00000}%
\definecolor{unino}{rgb}{0.00000,0.50000,0.50000}
\definecolor{purg}{rgb}{0.640000,0.00000,0.00000}%
\definecolor{no}{rgb}{1.00000,0.20000,0.00000}%

\begin{tikzpicture}

\begin{axis}[%
width=9cm,
height=7cm,
at={(0.5cm,0.5cm)},
scale only axis,
xmin=1,
xmax=1.7,
xtick={1,1.1,1.4,1.7},
xticklabels={{$1$},{$1.1$},{$1.4$},{$1.7$}},
scaled x ticks = false,
ymin=100,
ymax=10000,
ytick={100,333,1000,10000,100000},
yticklabels={{$\frac{0.1}{\lambda}$},{$\frac{1}{3\lambda}$},{$\frac{1}{\lambda}$},{$\frac{10}{\lambda}$},{$\frac{100}{\lambda}$}},
xlabel style={font=\color{white!15!black}},
xlabel={$\Omega$},
ymode=log,
ymin=100,
ymax=100000,
yminorticks=true,
xmajorgrids,
ymajorgrids,
ylabel style={font=\color{white!15!black}},
ylabel={Average in-order execution time},
axis background/.style={fill=white},
title style={font=\bfseries},
legend style={legend cell align=left, align=left, draw=white!15!black}
]
\addplot [color=ginipurg, mark=*,line width = 1.5pt, mark options={solid, fill=ginipurg, ginipurg}]
  table[row sep=crcr]{%
1	593.918109525146\\
1.07777777777778	500.292509525147\\
1.15555555555556	504.054509525146\\
1.23333333333333	505.303349525147\\
1.31111111111111	504.607899525146\\
1.38888888888889	502.911529525146\\
1.46666666666667	504.973389525146\\
1.54444444444444	502.308309525147\\
1.62222222222222	503.833329525147\\
1.7	502.321139525147\\
};
%\addlegendentry{gini with purging}

\addplot [color=unipurg, mark=triangle*, line width = 1.5pt, mark options={solid, fill=unipurg, unipurg}]
  table[row sep=crcr]{%
1	47175.2392495251\\
1.07777777777778	13129.6691695251\\
1.15555555555556	3761.78254952515\\
1.23333333333333	1705.25072952515\\
1.31111111111111	988.750659525146\\
1.38888888888889	734.183859525146\\
1.46666666666667	617.356199525146\\
1.54444444444444	559.654310306709\\
1.62222222222222	527.325319525147\\
1.7	509.394639525146\\
};
%\addlegendentry{uni with purging}

\addplot [color=purg,line width = 1.5pt, mark=+,mark size=3pt, mark options={solid, fill=purg, purg}]
  table[row sep=crcr]{%
1	816.259729525146\\
1.07777777777778	763.664609525146\\
1.15555555555556	705.709999525146\\
1.23333333333333	675.290949525146\\
1.31111111111111	645.782859525146\\
1.38888888888889	629.153299525147\\
1.46666666666667	615.261219525147\\
1.54444444444444	602.319679525147\\
1.62222222222222	603.665309525147\\
1.7	583.453759525147\\
};
%\addlegendentry{our with purging}

\node at (1.15,300) {\textcolor{ginipurg}{Ideal (purging)}};
\node at (1.15,1100) {\textcolor{purg}{Optimal (purging)}};
\node at (1.3,10000) {\textcolor{unipurg}{Uniform (purging)}};
\iffalse
\addplot [color=black, mark=*, mark options={solid, fill=black, black}]
  table[row sep=crcr]{%
1	505\\
1.07777777777778	514\\
1.15555555555556	518\\
1.23333333333333	523\\
1.31111111111111	529\\
1.38888888888889	527\\
1.46666666666667	535\\
1.54444444444444	540\\
1.62222222222222	545\\
1.7	547\\
};\fi

\end{axis}
\end{tikzpicture}%

%% file: main.bbl
\begin{thebibliography}{10}
\providecommand{\url}[1]{#1}
\csname url@samestyle\endcsname
\providecommand{\newblock}{\relax}
\providecommand{\bibinfo}[2]{#2}
\providecommand{\BIBentrySTDinterwordspacing}{\spaceskip=0pt\relax}
\providecommand{\BIBentryALTinterwordstretchfactor}{4}
\providecommand{\BIBentryALTinterwordspacing}{\spaceskip=\fontdimen2\font plus
\BIBentryALTinterwordstretchfactor\fontdimen3\font minus
  \fontdimen4\font\relax}
\providecommand{\BIBforeignlanguage}[2]{{%
\expandafter\ifx\csname l@#1\endcsname\relax
\typeout{** WARNING: IEEEtran.bst: No hyphenation pattern has been}%
\typeout{** loaded for the language `#1'. Using the pattern for}%
\typeout{** the default language instead.}%
\else
\language=\csname l@#1\endcsname
\fi
#2}}
\providecommand{\BIBdecl}{\relax}
\BIBdecl

\bibitem{moore1965cramming}
G.~E. Moore \emph{et~al.}, ``Cramming more components onto integrated
  circuits,'' 1965.

\bibitem{moore1965moore}
G.~Moore, ``Moore’s law,'' \emph{Electronics Magazine}, vol.~38, no.~8, p.
  114, 1965.

\bibitem{dutta2019short}
S.~Dutta, V.~Cadambe, and P.~Grover, ``“{S}hort-{D}ot”: Computing large
  linear transforms distributedly using coded short dot products,'' \emph{IEEE
  Transactions on Information Theory}, vol.~65, no.~10, pp. 6171--6193, 2019.

\bibitem{mallick2019rateless}
A.~Mallick, M.~Chaudhari, U.~Sheth, G.~Palanikumar, and G.~Joshi, ``Rateless
  codes for near-perfect load balancing in distributed matrix-vector
  multiplication,'' \emph{Proceedings of the ACM on Measurement and Analysis of
  Computing Systems}, vol.~3, no.~3, pp. 1--40, 2019.

\bibitem{yu2019lagrange}
Q.~Yu, S.~Li, N.~Raviv, S.~M.~M. Kalan, M.~Soltanolkotabi, and S.~A.
  Avestimehr, ``Lagrange coded computing: Optimal design for resiliency,
  security, and privacy,'' in \emph{The 22nd International Conference on
  Artificial Intelligence and Statistics}.\hskip 1em plus 0.5em minus
  0.4em\relax PMLR, 2019, pp. 1215--1225.

\bibitem{dean2013tail}
J.~Dean and L.~A. Barroso, ``The tail at scale,'' \emph{Communications of the
  ACM}, vol.~56, no.~2, pp. 74--80, 2013.

\bibitem{joshi2017efficient}
G.~Joshi, E.~Soljanin, and G.~Wornell, ``Efficient redundancy techniques for
  latency reduction in cloud systems,'' \emph{ACM Transactions on Modeling and
  Performance Evaluation of Computing Systems (TOMPECS)}, vol.~2, no.~2, p.~12,
  2017.

\bibitem{ramamoorthy2020straggler}
A.~Ramamoorthy, A.~B. Das, and L.~Tang, ``Straggler-resistant distributed
  matrix computation via coding theory: Removing a bottleneck in large-scale
  data processing,'' \emph{IEEE Signal Processing Magazine}, vol.~37, no.~3,
  pp. 136--145, 2020.

\bibitem{sheth2018application}
U.~Sheth, S.~Dutta, M.~Chaudhari, H.~Jeong, Y.~Yang, J.~Kohonen, T.~Roos, and
  P.~Grover, ``An application of storage-optimal matdot codes for coded matrix
  multiplication: Fast k-nearest neighbors estimation,'' in \emph{2018 IEEE
  International Conference on Big Data (Big Data)}.\hskip 1em plus 0.5em minus
  0.4em\relax IEEE, 2018, pp. 1113--1120.

\bibitem{d2020gasp}
R.~G. D’Oliveira, S.~El~Rouayheb, and D.~Karpuk, ``{GASP} codes for secure
  distributed matrix multiplication,'' \emph{IEEE Transactions on Information
  Theory}, vol.~66, no.~7, pp. 4038--4050, 2020.

\bibitem{wahabfederated}
O.~A. Wahab, A.~Mourad, H.~Otrok, and T.~Taleb, ``Federated machine learning:
  Survey, multi-level classification, desirable criteria and future directions
  in communication and networking systems.''

\bibitem{yu2017polynomialn}
Q.~Yu, M.~A. Maddah-Ali, and S.~Avestimehr, ``Polynomial codes: an optimal
  design for high-dimensional coded matrix multiplication,'' in \emph{NIPS},
  2017.

\bibitem{raviv2020gradient}
N.~Raviv, I.~Tamo, R.~Tandon, and A.~G. Dimakis, ``Gradient coding from cyclic
  {MDS} codes and expander graphs,'' \emph{IEEE Transactions on Information
  Theory}, vol.~66, no.~12, pp. 7475--7489, 2020.

\bibitem{sarikaya2019motivating}
Y.~Sarikaya and O.~Ercetin, ``Motivating workers in federated learning: A
  stackelberg game perspective,'' \emph{IEEE Networking Letters}, vol.~2,
  no.~1, pp. 23--27, 2019.

\bibitem{lim2020federated}
W.~Y.~B. Lim, N.~C. Luong, D.~T. Hoang, Y.~Jiao, Y.-C. Liang, Q.~Yang,
  D.~Niyato, and C.~Miao, ``Federated learning in mobile edge networks: A
  comprehensive survey,'' \emph{IEEE Communications Surveys \& Tutorials},
  vol.~22, no.~3, pp. 2031--2063, 2020.

\bibitem{feng2019joint}
S.~Feng, D.~Niyato, P.~Wang, D.~I. Kim, and Y.-C. Liang, ``Joint service
  pricing and cooperative relay communication for federated learning,'' in
  \emph{2019 International Conference on Internet of Things (iThings) and IEEE
  Green Computing and Communications (GreenCom) and IEEE Cyber, Physical and
  Social Computing (CPSCom) and IEEE Smart Data (SmartData)}.\hskip 1em plus
  0.5em minus 0.4em\relax IEEE, 2019, pp. 815--820.

\bibitem{9155442}
D.~{Malak}, A.~{Cohen}, and M.~{Médard}, ``How to distribute computation in
  networks,'' in \emph{IEEE INFOCOM 2020 - IEEE Conference on Computer
  Communications}, 2020, pp. 327--336.

\bibitem{duffy2021mds}
K.~R. Duffy and S.~Shneer, ``{MDS} coding is better than replication for job
  completion times,'' \emph{Operations Research Letters}, vol.~49, no.~1, pp.
  91--95, 2021.

\bibitem{li2016unified}
S.~Li, M.~A. Maddah-Ali, and A.~S. Avestimehr, ``A unified coding framework for
  distributed computing with straggling servers,'' in \emph{2016 IEEE Globecom
  Workshops (GC Wkshps)}.\hskip 1em plus 0.5em minus 0.4em\relax IEEE, 2016,
  pp. 1--6.

\bibitem{lee2017high}
K.~Lee, C.~Suh, and K.~Ramchandran, ``High-dimensional coded matrix
  multiplication,'' in \emph{2017 IEEE International Symposium on Information
  Theory (ISIT)}.\hskip 1em plus 0.5em minus 0.4em\relax IEEE, 2017, pp.
  2418--2422.

\bibitem{kim2019coded}
M.~Kim, J.-y. Sohn, and J.~Moon, ``Coded matrix multiplication on a group-based
  model,'' in \emph{2019 IEEE International Symposium on Information Theory
  (ISIT)}.\hskip 1em plus 0.5em minus 0.4em\relax IEEE, 2019, pp. 722--726.

\bibitem{lee2017speeding}
K.~Lee, M.~Lam, R.~Pedarsani, D.~Papailiopoulos, and K.~Ramchandran, ``Speeding
  up distributed machine learning using codes,'' \emph{IEEE Transactions on
  Information Theory}, vol.~64, no.~3, pp. 1514--1529, 2017.

\bibitem{tandon2016gradient}
R.~Tandon, Q.~Lei, A.~G. Dimakis, and N.~Karampatziakis, ``Gradient coding,''
  \emph{arXiv preprint arXiv:1612.03301}, 2016.

\bibitem{tandon2017gradient}
------, ``Gradient coding: Avoiding stragglers in distributed learning,'' in
  \emph{International Conference on Machine Learning}.\hskip 1em plus 0.5em
  minus 0.4em\relax PMLR, 2017, pp. 3368--3376.

\bibitem{raviv2018gradient}
N.~Raviv, R.~Tandon, A.~Dimakis, and I.~Tamo, ``Gradient coding from cyclic
  {MDS} codes and expander graphs,'' in \emph{International Conference on
  Machine Learning}.\hskip 1em plus 0.5em minus 0.4em\relax PMLR, 2018, pp.
  4305--4313.

\bibitem{Dutta2020}
S.~{Dutta}, M.~{Fahim}, F.~{Haddadpour}, H.~{Jeong}, V.~{Cadambe}, and
  P.~{Grover}, ``On the optimal recovery threshold of coded matrix
  multiplication,'' \emph{IEEE Transactions on Information Theory}, vol.~66,
  no.~1, pp. 278--301, 2020.

\bibitem{yu2017polynomial}
Q.~Yu, M.~A. Maddah-Ali, and A.~S. Avestimehr, ``Polynomial codes: an optimal
  design for high-dimensional coded matrix multiplication,'' \emph{arXiv
  preprint arXiv:1705.10464}, 2017.

\bibitem{baharav2018straggler}
T.~Baharav, K.~Lee, O.~Ocal, and K.~Ramchandran, ``Straggler-proofing
  massive-scale distributed matrix multiplication with d-dimensional product
  codes,'' in \emph{2018 IEEE International Symposium on Information Theory
  (ISIT)}.\hskip 1em plus 0.5em minus 0.4em\relax IEEE, 2018, pp. 1993--1997.

\bibitem{li2015coded}
S.~Li, M.~A. Maddah-Ali, and A.~S. Avestimehr, ``Coded mapreduce,'' in
  \emph{2015 53rd Annual Allerton Conference on Communication, Control, and
  Computing (Allerton)}.\hskip 1em plus 0.5em minus 0.4em\relax IEEE, 2015, pp.
  964--971.

\bibitem{attia2019near}
M.~A. Attia and R.~Tandon, ``Near optimal coded data shuffling for distributed
  learning,'' \emph{IEEE Transactions on Information Theory}, vol.~65, no.~11,
  pp. 7325--7349, 2019.

\bibitem{song2019pliable}
L.~Song, C.~Fragouli, and T.~Zhao, ``A pliable index coding approach to data
  shuffling,'' \emph{IEEE Transactions on Information Theory}, vol.~66, no.~3,
  pp. 1333--1353, 2019.

\bibitem{Klein1975}
L.~Kleinrock, \emph{Queuing Systems Vol. I: Theory}.\hskip 1em plus 0.5em minus
  0.4em\relax New York: Wiley, 1975.

\bibitem{BaharavISIT2018}
T.~{Baharav}, K.~{Lee}, O.~{Ocal}, and K.~{Ramchandran}, ``Straggler-proofing
  massive-scale distributed matrix multiplication with d-dimensional product
  codes,'' in \emph{IEEE International Symposium on Information Theory (ISIT)},
  2018, pp. 1993--1997.

\bibitem{gallager2013stochastic}
R.~G. Gallager, \emph{Stochastic processes: theory for applications}.\hskip 1em
  plus 0.5em minus 0.4em\relax Cambridge University Press, 2013.

\bibitem{bentley1975multidimensional}
J.~L. Bentley, ``Multidimensional binary search trees used for associative
  searching,'' \emph{Communications of the ACM}, vol.~18, no.~9, pp. 509--517,
  1975.

\bibitem{nowak2008generalized}
R.~Nowak, ``Generalized binary search,'' in \emph{2008 46th Annual Allerton
  Conference on Communication, Control, and Computing}.\hskip 1em plus 0.5em
  minus 0.4em\relax IEEE, 2008, pp. 568--574.

\bibitem{boyd2004convex}
S.~Boyd, S.~P. Boyd, and L.~Vandenberghe, \emph{Convex optimization}.\hskip 1em
  plus 0.5em minus 0.4em\relax Cambridge university press, 2004.

\end{thebibliography}
